\documentclass[twocolumn, 5p, sort, compress]{elsarticle}
\usepackage{amsmath, amssymb, amsfonts, amsthm}
\usepackage{bm, lipsum}
\usepackage{physics}
\usepackage{graphicx, enumitem}
\usepackage{color, xcolor, soul, framed}
\usepackage{cancel}
\usepackage{booktabs}
\usepackage{geometry}
\usepackage{hyperref}
\usepackage{mathtools}
\usepackage{makecell}

\newtheorem{proposition}{Proposition}
\newtheorem{theorem}{Theorem}

\newtheorem{corollary}{Corollary}
\newtheorem{definition}{Definition}
\newtheorem{remark}{Remark}

\usepackage{stfloats}
\usepackage{array}
\usepackage{diagbox}
\usepackage{multirow}

\usepackage{color, colortbl}
\definecolor{LightCyan}{rgb}{0.88,1,1}
\definecolor{Gray}{gray}{0.93}

\newcommand{\CC}{\cellcolor{Gray!80!white}}
\newcommand{\CB}{\cellcolor{green!30!white}}
\newcommand{\CG}{\cellcolor{brown!45!white}}
\newcommand{\CY}{\cellcolor{yellow!30!white}}

\usepackage{tikz}
\usetikzlibrary{arrows.meta, positioning, calc, shapes.geometric}
\usetikzlibrary{shapes.geometric}
\definecolor{silver}{HTML}{C0C0C0}
\definecolor{gold}{HTML}{FFD700}

\hypersetup{
    colorlinks=true,    citecolor=blue,    linkcolor=red,    filecolor=magenta,  urlcolor=blue,    pdftitle={Overleaf Example}}
    
\usepackage{fancyhdr}
\title{Revisiting Certainty Equivalence: The Structural Coupling Between Estimation and Control in Underactuated Nonlinear Systems}
\date{\today}

\begin{document}

\author[1]{Daniel Engelsman}
\ead{Dengelsm@campus.haifa.ac.il}

\affiliation[1]{organization={The Hatter Department of Marine Technologies, School of Marine Sciences University of Haifa}, city={Haifa}, country={Israel}.}
            
\author[1]{Itzik Klein}
\ead{kitzik@univ.haifa.ac.il}

  
\begin{abstract}
The certainty equivalence (CE) principle underpins a wide range of control architectures by enabling the separation of estimation and control design. While this property holds for linear systems, its validity in nonlinear settings remains limited and often implicitly assumed. This paper revisits CE from a nonlinear perspective, showing that estimated states induce an intrinsic coupling between estimation and tracking dynamics. By analyzing the closed-loop system in tracking-error coordinates, we demonstrate that nonlinear state dependence gives rise to higher-order interaction terms during aggressive transients.
Motivated by this limitation, we propose an \emph{estimation-aware} (EA) control paradigm that incorporates estimation quality into the feedback law to isolate estimation-induced loops. The formulation remains filtering-agnostic while preserving general applicability to smooth, underactuated nonlinear systems. We derive analytical conditions guaranteeing bounded tracking under uncertainty, validating the framework under high-fidelity quadrotor flight simulation along complex 3D trajectories at speeds up to 57.6~km/h. Frequency-domain evaluations demonstrate that the EA law extends tracking bandwidth by $39\%$ and improves stability margins by up to $55\%$, effectively mitigating severe cross-couplings to offer a robust alternative to classical CE-based designs.
\end{abstract}
\maketitle

\section{Introduction} 
The separation between state estimation and control design is a cornerstone of modern control theory. For linear time-invariant (LTI) systems, it enables independent design of observers and state-feedback controllers while preserving closed-loop (CL) stability. Under the classical Certainty Equivalence (CE) principle, the controller is constructed as if the estimated state were the true state, simplifying analysis and design, although estimation errors may still affect transient performance \cite{HESPANHA19991}.
For nonlinear systems, however, this decoupling generally breaks down. Because the dynamics and output mappings are explicitly state-dependent, estimation errors directly alter the CL trajectories and can compromise stability \cite{slotine1991applied}. Nevertheless, CE-based architectures remain ubiquitous in practice, often relying on the implicit assumption that estimation errors remain small or evolve on a significantly faster time scale \cite{isidori1985nonlinear,Molin2013}.
\\
While reasonable near equilibrium, these assumptions routinely fail during aggressive transients or persistent uncertainty, where limited or state-dependent observability can lead the controller to act on highly inaccurate state estimates \cite{grune2016nonlinear}. Several frameworks have historically sought to mitigate the breakdown of the separation principle in nonlinear settings. High-gain observer architectures isolate estimation dynamics by enforcing a fast, artificial time-scale separation, though they remain inherently susceptible to measurement noise and the notorious peaking phenomenon \cite{Hongwei2021,kim2019comprehensive}. 
\\
Alternatively, adaptive and dual control paradigms explicitly account for estimation uncertainty, but they frequently suffer from severe computational intractability or restrict applicability to specific parametric structures \cite{gu2025robust}. Similarly, while robust approaches like dynamic surface control or backstepping handle certain observer-induced errors, they generally view the estimation mismatch as an exogenous perturbation rather than an structural feedback loop  \cite{Jing2020,xu2020backstepping}. 
Consequently, a general, filtering-agnostic framework that treats the coupled estimator-controller interaction as an intrinsic CL property remains an open challenge \cite{COHEN2024103565,STEINERT2025103553}. 
\\
This paper revisits the CE principle from a nonlinear perspective, treating the estimator-controller interaction as an intrinsic CL property rather than an unmodeled perturbation. 
\\
By formulating the system in estimated tracking-error coordinates, we demonstrate that state-dependent nonlinearities induce non-negligible coupling terms that emerge through both first-order sensitivity to estimation errors and higher-order residuals. 
Using Lyapunov-based arguments, we derive analytical conditions under which the tracking error remains bounded, establishing a direct mathematical relationship between estimation accuracy and trackability. Lastly, we evaluate this framework through high-fidelity numerical simulations. 
\\
Ultimately, the results contrast with classical separation- or perturbation-based analyses by providing a direct assessment of performance degradation under imperfect state information. 
We thereby establish that the certainty equivalence breakdown is not merely a consequence of unmodeled dynamics or noise, but an inherent, fundamental property of estimate-based nonlinear control.
In summary, the primary contributions of this work are threefold:
\begin{enumerate}[label=\roman*), itemsep=1.5pt, parsep=1pt, topsep=2pt]
\item \textbf{Unified Formulation:} We formalize the intrinsic CL coupling between estimation and control in nonlinear architectures and derive analytical tracking bounds under estimation uncertainty.
\item \textbf{Empirical Characterization:} We implement a rigorous validation protocol using high-fidelity numerical simulations, providing a quantifiable benchmark for CL stability margins and tracking bandwidth extensions.
\item \textbf{Agnostic Architecture:} A modular, filter- and controller-agnostic framework is presented that is fully reproducible and openly shared on {\texttt{GitHub}}\footnote{\fontsize{8pt}{8pt}\selectfont\url{https://github.com/ANSFL/Estimation-Aware-Control}}.
\end{enumerate}
Given the rigorous demands of modern autonomy, these findings mark a significant advancement, achieving a 55\% stability margin improvement and a 39\% tracking bandwidth extension. 
\\
The remainder of the paper is structured as follows: Section~\ref{sec:theory} introduces the theoretical preliminaries, Section~\ref{sec:prob} formulates the problem, and Section~\ref{sec:stability} develops the stability analysis. Next, Section~\ref{sec:prop} presents the proposed strategy, Section~\ref{sec:results} reports the simulation results, and Section~\ref{sec:conc} concludes the paper.

\section{Preliminaries} \label{sec:theory}
This section establishes the mathematical notation, system class, and fundamental definitions utilized throughout this study.

\subsection{Nonlinear System Dynamics}
We consider a class of nonlinear, continuous-time, multiple-input multiple-output (MIMO) systems. Rooted in Newton-Euler rigid-body dynamics, these systems naturally manifest a control-affine structure of the form \cite{slotine1991applied}
\begin{align} \label{eq:state_space_GT}
\mathcal{S} : \left\{ 
\begin{aligned}
\dot{\boldsymbol{x}} &= \boldsymbol{f}(\boldsymbol{x}) + \mathbf{g}(\boldsymbol{x})\boldsymbol{u} + \boldsymbol{w} \, ,  \\
\boldsymbol{y} &= \boldsymbol{h} \big(\boldsymbol{x}\big) + \boldsymbol{v} \, ,
\end{aligned} \right.
\end{align}
where time arguments are omitted for brevity. Here, $\boldsymbol{x}$ denotes the hidden system states evolving on an $n$-dimensional smooth manifold $\mathcal{M}$, $\boldsymbol{u} \in \mathbb{R}^{m}$ is a known exogenous input, and $\boldsymbol{y} \in \mathbb{R}^{p}$ is the measurable output. The mapping $\boldsymbol{f} : \mathcal{M} \to T_x \mathcal{M}$ represents the drift vector field into the tangent space $T_x \mathcal{M}$, the input distribution matrix by $\mathbf{g} : \mathcal{M} \to \mathbb{R}^{n \times m}$, and $\boldsymbol{h} : \mathcal{M} \to \mathbb{R}^p$ describes the measurement map. The process and measurement disturbances, $\boldsymbol{w} \in \mathbb{R}^{n}$ and $\boldsymbol{v} \in \mathbb{R}^{p}$, are modeled as zero-mean white Gaussian processes with respective covariance matrices $\boldsymbol{\Sigma_w}$ and $\boldsymbol{\Sigma_v}$. 
%

\subsection{Nonlinear Tracking Formulation}
The control of complex dynamical systems generally evolves from idealized full-state assumptions to more practical output-feedback (OF) formulations. This section outlines the shift from global state stabilization to input-output (I/O) linearization via dynamic inversion.

\textbf{State-feedback (SF):} In an ideal framework where the complete state vector $\boldsymbol{x} \in \mathbb{R}^n$ is measurable, we define the tracking error relative to a smooth, time-parameterized reference $\boldsymbol{x}_{\mathrm{ref}}(t)$ as
\begin{align} \label{eq:epsilon_x}
\boldsymbol{\epsilon}_{\boldsymbol{x}} = {\boldsymbol{x}} - \boldsymbol{x}_{\mathrm{ref}} \ \in \ \mathbb{R}^{n} \, .
\end{align}
For both regulation ($\dot{\boldsymbol{x}}_{\mathrm{ref}} = \mathbf{0}$) or tracking ($\dot{\boldsymbol{x}}_{\mathrm{ref}} \neq \mathbf{0}$) tasks, the error dynamics obey the nonlinear vector field
\begin{align} \label{eq:epsilon_x_dot}
\dot{\boldsymbol{\epsilon}}_{\boldsymbol{x}} = \boldsymbol{f} ( {\boldsymbol{x}} ) + \mathbf{g}(\boldsymbol{x}) \boldsymbol{u} - \dot{\boldsymbol{x}}_{\mathrm{ref}} \, .
\end{align}
While optimization-based strategies like model predictive control (MPC) can solve this directly, they often incur prohibitive computational costs in real-time applications \cite{Hongwei2021,Sihao2022}. A common algebraic alternative is feedback linearization (FL), also referred to as dynamic inversion, by prescribing a desired linear evolution $\dot{\boldsymbol{\epsilon}}_{\boldsymbol{x}}^\star \triangleq -\mathbf{K} \boldsymbol{\epsilon}_{\boldsymbol{x}}$, where $\mathbf{K} \succ 0$ is a positive-definite (PD) gain matrix. The resulting control input is thus
\begin{align} \label{eq:control_law_GT}
\boldsymbol{u} = \mathbf{g}^{\dagger} (\boldsymbol{x}) \bigl( \dot{\boldsymbol{\epsilon}}_{\boldsymbol{x}}^\star + \dot{\boldsymbol{x}}_{\mathrm{ref}} - \boldsymbol{f}(\boldsymbol{x}) \bigr) \, ,
\end{align}
where $(\cdot)^\dagger$ denotes the Moore-Penrose pseudoinverse. While elegant, the SF approach encounters two critical hurdles: (i) observability constraints, as the full state $\boldsymbol{x}$ is rarely available through direct sensing, and (ii) structural underactuation ($m<n$), which renders $\mathbf{g}(\boldsymbol{x})$ non-invertible such that its image does not span the entire state manifold. 

\textbf{Output-feedback (OF):} As these constraints preclude full-state linearization, they necessitate a formulation concentrated on a lower-dimensional task space \cite{Jing2020,wang2019stability}. Let the output tracking error be defined as
\begin{align} \label{eq:epsilon_y}
\boldsymbol{\epsilon}_{\boldsymbol{y}} = {\boldsymbol{y}} - \boldsymbol{y}_{\mathrm{ref}} \ \in \ \mathbb{R}^{p} \, , \quad p \le m <n \, .
\end{align}
Because the control input $\boldsymbol{u}$ is typically latent in lower-order output derivatives, successive differentiation is required to expose the control authority. Based on the Byrnes-Isidori normal form \cite{isidori1985nonlinear}, the derivatives of the output are obtained via the recursive Lie mapping
\begin{align} \label{eq:y_k_mapping}
\boldsymbol{y}^{(k)} = {L}_f^{k} \boldsymbol{h} (\boldsymbol{x}) \, , \quad \forall \quad k < r \, .
\end{align}
Here, $r$ denotes the relative degree—defined as the smallest integer for which the control input appears explicitly in the $r$-th derivative $\boldsymbol{y}^{(r)}$. This relationship is established as
\begin{align} \label{eq:normal_form}
\boldsymbol{y}^{(r)} = \underbrace{ {L}_f^{r} \boldsymbol{h}(\boldsymbol{x})}_{\boldsymbol{\alpha}(\boldsymbol{x})} + \underbrace{{L}_\text{g} {L}_f^{r-1} \boldsymbol{h}(\boldsymbol{x})}_{\boldsymbol{\beta}(\boldsymbol{x})} \boldsymbol{u} \, ,
\end{align}
where ${L}_{\boldsymbol{f}}$ and ${L}_{\mathbf{g}}$ denote the Lie derivatives along the vector fields $\boldsymbol{f}$ and $\mathbf{g}$, respectively.

\begin{remark}
While output derivatives $\boldsymbol{y}^{(0:r-1)}$ are governed by plant physics, the reference $\boldsymbol{y}_{\mathrm{ref}}(t)$ is exogenous and independent of the system dynamics. To avoid algebraic singularities, the reference must be $C^r$ continuous—a requirement naturally satisfied by modern trajectory generators like minimum-snap or flatness-based planners \cite{Ryan2007,Hagenmeyer2010}
\begin{align} \label{eq:ref_diff}
\boldsymbol{y}_{\mathrm{ref}}^{(0:r)}(t) \triangleq \left\{ \frac{d^k}{dt^k} \, \boldsymbol{y}_{\mathrm{ref}}(t) \right\}_{k=0}^r \, .
\end{align} 
\end{remark}
To enforce exponential convergence of the tracking error, we prescribe the desired closed-loop dynamics
\begin{align} \label{eq:error_dot_y}
\boldsymbol{\epsilon}_{\boldsymbol{y}}^{(r)} + \mathbf{K}_{r-1} \boldsymbol{\epsilon}_{\boldsymbol{y}}^{(r-1)} + \cdots + \mathbf{K}_0 \boldsymbol{\epsilon}_{\boldsymbol{y}}^{(0)} = \mathbf{0} \, .
\end{align}
This objective is achieved by defining a virtual control input $\boldsymbol{\nu} \triangleq \boldsymbol{y}^{(r)}$ such that
\begin{align} \label{eq:virtual_inp}
\boldsymbol{\nu} \triangleq \boldsymbol{y}_{\mathrm{ref}}^{(r)} - \sum_{j=0}^{r-1} \mathbf{K}_j \underbrace{ \left( \boldsymbol{y}^{(j)} - \boldsymbol{y}^{(j)}_{\mathrm{ref}} \right) }_{\boldsymbol{\epsilon_y}^{(j)}} \, ,
\end{align}
where $\mathbf{K}_j \succ \mathbf{0}$ is chosen such that the resulting characteristic polynomial is Hurwitz. 

\textbf{Dynamic inversion:} The mapping between the virtual input $\boldsymbol{\nu}$ and the physical control command $\boldsymbol{u}$ is established via dynamic inversion. Provided the decoupling matrix $\boldsymbol{\beta}(\boldsymbol{x})$ is non-singular, the input-output relationship in \eqref{eq:normal_form} can be inverted to yield
\begin{align} \label{eq:dynamic_inversion_1}
\boldsymbol{u} = \boldsymbol{\beta}^{-1}(\boldsymbol{x}) \left( \boldsymbol{\nu} - \boldsymbol{\alpha}(\boldsymbol{x}) \right) \, .
\end{align}
Under exact model knowledge, this law algebraically cancels plant nonlinearities, yielding a linearized closed-loop system of decoupled $r$-th order integrators
\begin{align} \label{eq:dynamic_inversion_2}
\boldsymbol{y}^{(r)} = \boldsymbol{\nu} \, .
\end{align}
\begin{remark}
For cases where $r < n$, the tracking of $\boldsymbol{y}$ does not guarantee the stability of the full state $\boldsymbol{x}$. One must verify that the $n-r$ internal dynamics are minimum phase to prevent internal state divergence during output regulation \cite{Ryan2007,Hagenmeyer2010,yadav2025control,engelsman2025lqg}.
\end{remark}

\subsection{Nonlinear State Observer}
Since $\boldsymbol{x}$ is not directly accessible, an estimate $\hat{\boldsymbol{x}} \in \mathcal{M}$ is constructed using the available inputs and measurements. Drawing from structure \eqref{eq:state_space_GT}, the nonlinear observer dynamics follow \cite{kay1993fundamentals}
\begin{align} \label{eq:state_space_est}
\widehat{\mathcal{S}} : \left\{ 
\begin{aligned}
\dot{\hat{\boldsymbol{x}}} &= \boldsymbol{f}(\hat{\boldsymbol{x}}) + \textbf{g}(\hat{\boldsymbol{x}}) \boldsymbol{u} + \boldsymbol{\ell} ( \hat{\boldsymbol{x}}, {\boldsymbol{y}} ) \, ,
\\
\hat{\boldsymbol{y}} &= \boldsymbol{h} \big(\hat{\boldsymbol{x}}\big) \, .
\end{aligned} \right.
\end{align}
The prediction phase evolves the estimate along vector field $\boldsymbol{f}(\, \cdot \,)$ and control distribution $\mathbf{g}(\, \cdot \,)$, while the correction employs the innovation term $\boldsymbol{\ell} (\, \cdot \,)$ to project the residuals into the tangent space $T_{\hat{x}} \mathcal{M}$.
Owing to process disturbances, the trajectory of $\widehat{\mathcal{S}}$ typically deviates from $\mathcal{S}$, inducing estimation and output residuals as
\begin{align}
\tilde{\boldsymbol{x}} \coloneqq \boldsymbol{x} - \hat{\boldsymbol{x}} \, ,  \label{eq:x_tilde} \\
\tilde{\boldsymbol{y}} \coloneqq \boldsymbol{y} - \hat{\boldsymbol{y}} \, . \label{eq:y_tilde}
\end{align}

\subsubsection{Linear Stochastic Estimation} While nonlinear estimation spans a broad class of global methods, the present analysis focuses on local tangent estimators, which linearize $\mathcal{M}$ around a transient operating point \cite{Bacon2000control}. For a control-affine system \eqref{eq:state_space_GT}, the first-order local behavior is characterized by linearizing the vector field about a nominal state-input pair $(\boldsymbol{x}_0,\boldsymbol{u}_0)$, yielding \cite{bar2001estimation}
\begin{align} \label{eq:linearize}
\begin{split}
\dot{\boldsymbol{x}} \approx & \boldsymbol{f}(\boldsymbol{x}_0) + \mathbf{g}(\boldsymbol{x}_0)\boldsymbol{u}_0 + \frac{\partial \boldsymbol{f}}{\partial \boldsymbol{x}}\Big|_{\boldsymbol{x}_0} \tilde{\boldsymbol{x}} 
\\
& + \frac{\partial (\mathbf{g}(\boldsymbol{x})\boldsymbol{u})}{\partial \boldsymbol{x}}\Big|_{(\boldsymbol{x}_0,\boldsymbol{u}_0)} \tilde{\boldsymbol{x}} + \frac{\partial (\mathbf{g}(\boldsymbol{x})\boldsymbol{u})}{\partial \boldsymbol{u}}\Big|_{(\boldsymbol{x}_0,\boldsymbol{u}_0)} \tilde{\boldsymbol{u}} \, .
\end{split}
\end{align}
In the state estimation context, the control input is a known deterministic signal, implying $\tilde{\boldsymbol{u}} = \boldsymbol{u} - \hat{\boldsymbol{u}} = \mathbf{0}$. Accordingly, no stochastic input uncertainty enters the covariance propagation, and the linearized dynamics evaluated at the current state estimate reduce to
\begin{align} \label{eq:jacobians}
\boldsymbol{F}_{\hat{\boldsymbol{x}}} \triangleq \left. \left( \dfrac{\partial \boldsymbol{f}}{\partial \boldsymbol{x}} + \dfrac{\partial \textbf{g}}{\partial \boldsymbol{x}} \boldsymbol{u} \right) \right|_{(\hat{\boldsymbol{x}}, \boldsymbol{u})} ,
\quad
\boldsymbol{H}_{\hat{\boldsymbol{x}}} \triangleq \left. \dfrac{\partial \boldsymbol{h}}{\partial \boldsymbol{x}} \right|_{\hat{\boldsymbol{x}}} ,
\end{align}
where $\boldsymbol{F}_{\hat{\boldsymbol{x}}}$ and $\boldsymbol{H}_{\hat{\boldsymbol{x}}}$, are the process and the measurement Jacobian matrices. In Newton–Euler models, inputs are typically state-independent, motivating the approximation $\boldsymbol{G}_{\hat{\boldsymbol{x}}} \approx \mathbf{g}(\hat{\boldsymbol{x}})$. Retaining only first-order perturbations, the nonlinear physics are approximated as
\begin{align}
\begin{split}
\boldsymbol{f} \big( \boldsymbol{x} \big) + \textbf{g} \big( \boldsymbol{x} \big) \boldsymbol{u} =& \boldsymbol{f} \big( \hat{\boldsymbol{x}} \big) + \textbf{g} \big( \hat{\boldsymbol{x}} \big) \boldsymbol{u} + \boldsymbol{F}_{\hat{\boldsymbol{x}}} \, \tilde{\boldsymbol{x}} + \mathcal{O} \left( \| \tilde{\boldsymbol{x}}  \|^2 \right) \, ,     
\end{split} \label{eq:error_state_1}
\\
\boldsymbol{h} \big( \boldsymbol{x} \big) =& \, \boldsymbol{h} \big( \hat{\boldsymbol{x}} \big) + \boldsymbol{H}_{\hat{\boldsymbol{x}}} \, \tilde{\boldsymbol{x}} + \mathcal{O} \left( \| \tilde{\boldsymbol{x}} \|^2 \right) \, .
\end{align}
For the observer gain, expanding about the current estimate and predicted measurement $(\hat{\boldsymbol{x}}, \hat{\boldsymbol{y}})$ yields
\begin{align}
\boldsymbol{\ell}(\hat{\boldsymbol{x}}, \boldsymbol{y}) \approx 
\boldsymbol{\ell}(\hat{\boldsymbol{x}}, \hat{\boldsymbol{y}} ) + \left. \frac{\partial \boldsymbol{\ell}}{\partial \hat{\boldsymbol{x}}} \right|_{(\hat{\boldsymbol{x}}, \hat{\boldsymbol{y}})} \tilde{\boldsymbol{x}} + \underbrace{ \left. \frac{\partial \boldsymbol{\ell}}{\partial \boldsymbol{y}} \right|_{(\hat{\boldsymbol{x}}, \hat{\boldsymbol{y}})} }_{ \boldsymbol{L}_{\hat{\boldsymbol{x}}} } \tilde{\boldsymbol{y}} \, .
\end{align}
By the consistency condition $\boldsymbol{\ell}(\hat{\boldsymbol{x}}, \boldsymbol{h}(\hat{\boldsymbol{x}})) = \mathbf{0}$, the first term vanishes, while the second is nullified as the expansion is evaluated at the current operating point where, by definition, $\tilde{\boldsymbol{x}} = \mathbf{0}$. 
Consequently, the correction reduces to a linear injection of the innovation
\begin{align} \label{eq:observer_approx}
\boldsymbol{\ell}(\hat{\boldsymbol{x}}, \boldsymbol{y}) = \boldsymbol{L}_{\hat{\boldsymbol{x}}} \tilde{\boldsymbol{y}} + \mathcal{O} \left( \| \tilde{\boldsymbol{y}} \|^2 \right) \, ,
\end{align}
where $\boldsymbol{L}_{\hat{\boldsymbol{x}}}$ is the state-dependent Kalman gain synthesized via a Differential Riccati Equation (DRE). 

\textbf{Error State-Space:} 
By combining the approximations in Eqs. \eqref{eq:error_state_1}–\eqref{eq:observer_approx}, we define the error system
\begin{align}
\widetilde{\mathcal{S}} \coloneqq \mathcal{S} - \widehat{\mathcal{S}} \, .
\end{align}
This describes the mismatch between the true plant \eqref{eq:state_space_GT} and the estimator \eqref{eq:state_space_est}. Upon linearization, the resulting perturbation dynamics become
\begin{align} \label{eq:error_space}
\widetilde{\mathcal{S}} : \left\{ 
\begin{aligned}
\dot{\tilde{\boldsymbol{x}}} &= \dot{\boldsymbol{x}} - \dot{\hat{\boldsymbol{x}}} = \boldsymbol{F}_{\hat{\boldsymbol{x}}} \, \tilde{\boldsymbol{x}} + \boldsymbol{w} - \boldsymbol{L}_{\hat{\boldsymbol{x}}} \tilde{\boldsymbol{y}} \, , 
\\
\tilde{\boldsymbol{y}} &= \boldsymbol{y} - \hat{\boldsymbol{y}} = \boldsymbol{H}_{\hat{\boldsymbol{x}}} \, \tilde{\boldsymbol{x}} +
\boldsymbol{v} \, .
\end{aligned} \right.
\end{align}
%
%
The synthesis of $\boldsymbol{L}_{\hat{\boldsymbol{x}}}$ highlights the intrinsic coupling between the estimation error $\tilde{\boldsymbol{x}}$ and the observer's corrective capacity. Unlike LTI systems, where the gain and the error dynamics are decoupled via the separation principle \cite{khalil1996robust}, the gain in \eqref{eq:observer_approx} is parameterized through Jacobians \eqref{eq:jacobians}, evaluated at $\hat{\boldsymbol{x}}$. Thereby, higher-order residuals $\mathcal{O}(\|\tilde{\boldsymbol{x}}\|^2)$ induce a state-dependent structural mismatch in the error system matrix $(\boldsymbol{F}_{\hat{\boldsymbol{x}}} - \boldsymbol{L}_{\hat{\boldsymbol{x}}} \boldsymbol{H}_{\hat{\boldsymbol{x}}})$.

\textbf{Bifurcated Tracking Error:}
The ideal tracking error $\boldsymbol{\epsilon}$, as defined in \eqref{eq:epsilon_x}, assumes full-state accessibility. However, in practical applications where $\boldsymbol{x}$ is not directly observable, the control law must be synthesized using a perceived tracking error 
\begin{align} \label{eq:err_perc}
\hat{\boldsymbol{\epsilon}} = \hat{\boldsymbol{x}} - \boldsymbol{x}_{\mathrm{ref}} \, .
\end{align}
To ensure uninterrupted feedback, the state estimator operates across two distinct regimes. Under nominal conditions, the estimator provides a CL estimate, $\hat{\boldsymbol{x}}^+$, by assimilating measurement updates to bound the estimation error. Conversely, during sensor outages or telemetry degradation, the system reverts to open-loop (OL) estimation, $\hat{\boldsymbol{x}}^-$, where the state is propagated via pure prediction. This duality manifests in the feedback loop as a bifurcated tracking error \cite{Engelsman2023}
\begin{align}
\hat{\boldsymbol{\epsilon}}^{\pm} = \hat{\boldsymbol{x}}^{\pm} - \boldsymbol{x}_{\mathrm{ref}} \, .
\end{align}

\section{Problem Formulation} \label{sec:prob}
With the system preliminaries established, this section examines the explicit coupling between estimation and control by embedding the error state-space $\widetilde{\mathcal{S}}$ into the tracking formulation. Under full-state SF settings, the perceived tracking error \eqref{eq:err_perc} transforms into
\begin{align}
\hat{\boldsymbol{\epsilon}}_{\boldsymbol{x}} = \hat{\boldsymbol{x}} - \boldsymbol{x}_{\mathrm{ref}} \overset{\eqref{eq:x_tilde}}{\longrightarrow} \boldsymbol{\epsilon}_{\boldsymbol{x}} - \tilde{\boldsymbol{x}} \, ,
\end{align}
where $\boldsymbol{\epsilon}_{\boldsymbol{x}}$ denotes the true (but unobservable) tracking error. The resulting dynamics are governed by
\begin{align} \label{eq:perceived_track_x}
\dot{\hat{\boldsymbol{\epsilon}}}_{\boldsymbol{x}} = \underbrace{ \boldsymbol{f} ( \hat{\boldsymbol{x}} ) + \mathbf{g}(\hat{\boldsymbol{x}}) \boldsymbol{u} - \dot{\boldsymbol{x}}_{\mathrm{ref}} }_{\text{Nominal tracking \eqref{eq:epsilon_x_dot}}} \ + \ \underbrace{\boldsymbol{L}_{\hat{\boldsymbol{x}}} \, \tilde{\boldsymbol{y}} }_{\text{\shortstack{Estimation\\coupling \eqref{eq:observer_approx}}}} \, .
\end{align}
Similarly, under the OF regime ($\dim(\boldsymbol{y}) < \dim(\boldsymbol{x})$), the perceived output tracking error evolves as
\begin{align}
\hat{\boldsymbol{\epsilon}}_{\boldsymbol{y}} = \hat{\boldsymbol{y}} - \boldsymbol{y}_{\mathrm{ref}} \overset{\eqref{eq:y_tilde}}{\longrightarrow} \boldsymbol{\epsilon}_{\boldsymbol{y}} - \tilde{\boldsymbol{y}} \, .
\end{align}
Applying the chain rule to the measurement model via the Jacobians \eqref{eq:jacobians}, the resulting tracking dynamics are
\begin{align} \label{eq:perceived_track_y}
\dot{\hat{\boldsymbol{\epsilon}}}_{\boldsymbol{y}} = 
\underbrace{ \boldsymbol{H}_{\hat{\boldsymbol{x}}} \left( \boldsymbol{f} ( \hat{\boldsymbol{x}} ) + \mathbf{g}(\hat{\boldsymbol{x}})\boldsymbol{u} \right) - \dot{\boldsymbol{y}}_{\mathrm{ref}} }_{ \text{\shortstack{Nominal output \\tracking \eqref{eq:state_space_est}}}} + \underbrace{ \boldsymbol{H}_{\hat{\boldsymbol{x}}} \boldsymbol{L}_{ \hat{\boldsymbol{x}}} \, \tilde{\boldsymbol{y}}
}_{ \text{\shortstack{Innovation\\coupling \eqref{eq:observer_approx}}}} + \mathcal{O}( \| \tilde{\boldsymbol{x}} \|^2) \, .
\end{align}
These decompositions reveal that $\dot{\hat{\boldsymbol{\epsilon}}}$ do not evolve purely along a physical vector field $\boldsymbol{f}$; instead, the control law must actively contend with innovation-based disturbances introduced by the estimator.

\begin{definition} \label{def:CE}
Let the state uncertainty be characterized by $\tilde{\boldsymbol{x}}$, with a PD covariance matrix $\mathbf{P} \triangleq \mathbb{E} [\tilde{\boldsymbol{x}} \tilde{\boldsymbol{x}}^\top] \succ \mathbf{0}$. The CE principle asserts that the control policy $\boldsymbol{\kappa}(\cdot)$ can be synthesized under the assumption of perfect state knowledge ($\hat{\boldsymbol{x}} \equiv \boldsymbol{x}$). Accordingly, the implemented control law relates to the nominal full-state policy via \cite{HESPANHA19991,Molin2013}
\begin{align} \label{eq:CE_equiv}
\boldsymbol{u} = \boldsymbol{\kappa}(\hat{\boldsymbol{\epsilon}}) \overset{\text{CE}}{\simeq} \boldsymbol{\kappa}(\boldsymbol{\epsilon}) \, .
\end{align}
\end{definition}

\subsection{The Deterministic Baseline}
The CE paradigm is widely adopted in nonlinear control as it decouples estimation from control synthesis \cite{grune2016nonlinear,yadav2025control}. Under this idealization, the closed-loop system is analyzed assuming exact model knowledge, perfect state reconstruction, and absence of stochastic disturbances. 
Neglecting the stochasticity injected by the state estimator implicitly adopts an "idealized CE" that rarely holds under high-dynamic conditions.

\textbf{Estimation-Induced Distortion:}
In practical tracking, idealized conditions are seldom realized; the perceived error tends to deviate from the true state in amplitude, phase, and latency, leading to the inequality $\hat{\boldsymbol{\epsilon}} \not \equiv \boldsymbol{\epsilon}$. Furthermore, during aggressive maneuvers, these terms intensify, contaminating $\hat{\boldsymbol{\epsilon}}$ and collapsing the structural separation. 
This sensitivity is exacerbated by Lie derivative-based linearization, which recursively amplifies estimation noise through higher-order derivatives until innovation coupling and residuals dominate the error budget. This deforms the underlying modeling assumptions, manifesting as state-dependent mismatches $\Delta \boldsymbol{f}$ and $\Delta \mathbf{g}$ that erode the controller's stability margins.
\begin{figure}[b]
\centering
\begin{tikzpicture}[>=Latex,
node distance=16mm and 24mm,
block/.style = {draw, rounded corners=3pt, thick, align=center, inner sep=4pt, minimum width=14mm, minimum height=10mm},
labelnode/.style={font=\normalsize, inner sep=1pt},
every node/.style={font=\normalsize},  
line_ax/.style={-{Stealth[length=5pt, width=5pt]}, line width=1.65pt},
line_er/.style={<->, >={Stealth[length=4pt, width=4pt]}, line width=0.8pt},
line_f_x/.style={yellow!95!white, line width=3.2pt},
line_f_r/.style={blue!45!cyan, line width=3.pt},
line_df/.style={green!50!white, line width=3.2pt},
d_O/.style={orange!90!yellow, line width=3.2pt},
d_line/.style={line width=1.pt, dashed},
frame/.style={line width=4pt, draw=gray!50},
circ/.style={thick, inner sep=0pt},
scale=\linewidth/0.8\textwidth]
]

\def\Lx{0}      \def\Ly{0}
\coordinate (O)     at (\Lx,\Ly);
\draw[line_ax] (O) -- ++(0,6);
\draw[line_ax] (O) -- ++(7.85,0);
\node at ($(O)+(-1.1, 6)$){};
\node at ($(O)+(6, -0.5)$){};

\node[circ] at (3.5, -0.8) (x_true){};
\node[circ] at ($(x_true)+(0, 3.15)$) (f_x_true){};
\node[circ] at ($(x_true)+(2, 5.5)$) (f_x_ref){};
\node[circ] at ($(x_true)+(-2.0, 1.75)$) (f_x_hat){};
\node[circ] at ($(f_x_true)+(0,-0.7)$) (f_x_dyn){};
\node[circ] at ($(O)+(7.95, -0.3)$){$\boldsymbol{\zeta}$};
\node[circ] at ($(O)+(-0.85, 5.75)$){$\boldsymbol{f(\zeta})$};
\node[circ] at ($(O)+(10, 0)$){};

\draw[line_f_x]($(O)+(0.7, .1)$) .. controls ($(f_x_hat)+(0,0)$) and ($(f_x_true)+(-.85,-0.4)$) .. ($(f_x_ref)+(0.6,-2.8)$);
\draw[line_f_r]($(O)+(0.08, 3.5)$) .. controls ($(f_x_hat)+(2,2.2)$) and ($(f_x_true)+(1.2,1.75)$) .. ($(f_x_ref)+(0.6,0.48)$);
\draw[line_df]($(O)+(0.25, 0.12)$) --++ (4.2, 2.9);
\draw[d_O] (f_x_dyn) --++ (f_x_true){}; 

\draw[d_line] (f_x_hat) -- ++(0, -3.3);
\draw[d_line] (f_x_hat) -- ++(8, 0);
\draw[d_line] (f_x_true) -- ++(0, -3.92);
\draw[d_line] (f_x_true) -- ++(4.9, 0);
\draw[d_line] (f_x_ref) -- ++(0, -7.05);
\draw[d_line] (f_x_ref) -- ++(4, 0);

\node[circ, fill=white, fill opacity=1] at (f_x_hat) {\large$\circ$};
\node[circ, fill=white, fill opacity=1] at (f_x_true) {\large$\circ$};
\node[circ, fill=white, fill opacity=1] at (f_x_ref) {\large$\circ$};
\node[circ, fill=white, fill opacity=1] at (f_x_dyn) {\tiny $\boldsymbol{\otimes}$};

\node[fill=white, fill opacity=1, inner sep=2.5pt] at ($(x_true)+(0,0)$) {$\boldsymbol{x}$};
\node[fill=white, fill opacity=1, inner sep=2pt] at ($(f_x_hat)+(0, -1.68)$) {$\hat{\boldsymbol{x}}$};
\node[fill=white, text opacity=1, fill opacity=1.0] at ($(f_x_ref)+(0.3, -5.55)$) {${\boldsymbol{x}}_{\mathrm{ref}}$};

\node[fill=white, fill opacity=1, inner sep=1pt] at ($(f_x_hat)+(5.5, 0)$) {$\boldsymbol{f(\hat{x})}$};
\node[fill=white, fill opacity=1, inner sep=1pt] at ($(f_x_true)+(3.5, 0)$) {$\boldsymbol{f({x})}$};
\node[fill=white, fill opacity=1, inner sep=1pt] at ($(f_x_ref)+(1.55, -.05)$) {$\dot{\boldsymbol{x}}_{\mathrm{ref}}$};

\draw[line_er] ($(f_x_hat)+(0.05,-2.5)$) -- ++(2-0.1, 0);
\draw[line_er] ($(f_x_hat)+(2.05,-2.5)$) -- ++(2-0.1, 0);
\draw[line_er] ($(f_x_hat)+(0.05,-3.3)$) -- ++(4-0.05, 0);
\node[circ, fill=white, fill opacity=1] at ($(f_x_hat)+(1.0,-2.43)$) {$\, \tilde{\boldsymbol{x}} \, $};
\node[circ, fill=white, fill opacity=1] at ($(f_x_hat)+(3.0,-2.49)$) {$\, \boldsymbol{\epsilon} \, $};
\node[circ, fill=white, fill opacity=1, inner sep=2pt] at ($(f_x_hat)+(2.0,-3.2)$) {$\, \hat{\boldsymbol{\epsilon}} \, $};

\draw[line_er] ($(f_x_true)+(4.9,0.05)$) -- ++(0.0, 2.22);
\draw[line_er] ($(f_x_hat)+(6.9,0.05)$) -- ++(0.0, 1.32);
\draw[line_er] ($(f_x_hat)+(7.9,0.05)$) -- ++(0.0, 3.6);
\node[circ, fill=white, fill opacity=1, inner sep=1.5pt] at ($(f_x_hat)+(6.9,0.71)$) {$\, \dot{\tilde{\boldsymbol{x}}} \, $};
\node[circ, fill=white, fill opacity=1, inner sep=1.5pt] at ($(f_x_hat)+(6.9,2.55)$) {$\, \dot{\boldsymbol{\epsilon}} \, $};
\node[circ, fill=white, fill opacity=1, inner sep=1.5pt] at ($(f_x_hat)+(7.9,1.85)$) {$\, \dot{\hat{\boldsymbol{\epsilon}}} \, $};
\end{tikzpicture}
\caption{In the nonlinear vector field $\boldsymbol{f}(\boldsymbol{\zeta})$, the platform tracks a reference trajectory (blue) to drive $\boldsymbol{\epsilon} \to \mathbf{0}$. Since the true states are latent (yellow), they are approximated by linearizing the dynamics about $\hat{\boldsymbol{x}}$ (green). This first-order approximation yields a second-order model error $\mathcal{O}(\| \tilde{\boldsymbol{x}} \|^2)$ (orange) relative to the true nonlinear dynamics ({\footnotesize$\boldsymbol{\otimes}$}).}
\label{fig:diagram}
\end{figure}
Fig.~\ref{fig:diagram} illustrates this divergence: the true latent state evolves along a stochastic manifold (yellow) toward the target reference (blue). 
While nominal dynamics (green) are propagated via $\boldsymbol{f}(\hat{\boldsymbol{x}})$, linearization residuals (orange) intensify during high-bandwidth maneuvers. In such high-intensity regimes, first-order approximations lose fidelity as the higher-order terms $\mathcal{O}(\|\tilde{\boldsymbol{x}}\|^2)$ become non-negligible.

\textbf{Control Augmentation:} To mitigate the influence of estimation-induced transients on the perceived tracking dynamics \eqref{eq:perceived_track_x}, the ideal inversion-based law \eqref{eq:control_law_GT} is augmented with an innovation compensation term as
\begin{align} \label{eq:control_law_est}
\boldsymbol{u} = \mathbf{g}^{\dagger} (\hat{\boldsymbol{x}}) \biggl( \underbrace{\dot{\boldsymbol{x}}_{\mathrm{ref}} }_{\mathrm{FF}}
- \underbrace{ \mathbf{K} \hat{\boldsymbol{\epsilon}}_{\boldsymbol{x}}  }_{\mathrm{FB}} - \underbrace{ \boldsymbol{f}(\hat{\boldsymbol{x}})  }_{\mathrm{MC}} - \underbrace{ \boldsymbol{L}_{\hat{\boldsymbol{x}}} \, \tilde{\boldsymbol{y}}  }_{\mathrm{IC}} \biggr) \, ,
\end{align}
comprising feedforward (FF), feedback (FB), model cancellation (MC), and innovation compensation (IC) elements. 
Similarly, solving for \eqref{eq:perceived_track_y} is addressed via the I/O linearization scheme \eqref{eq:dynamic_inversion_1}. This formulation exhibits a heightened sensitivity to $\hat{\boldsymbol{x}}$ due to the innovation coupling internalized within $\hat{\boldsymbol{ \epsilon}}_{\boldsymbol{y}}$. 
Accordingly, the virtual input mapping from \eqref{eq:virtual_inp} is synthesized as
\begin{align}
\boldsymbol{u} = \boldsymbol{\beta}^{\dagger}(\hat{\boldsymbol{x}}) \biggl( 
\underbrace{ \boldsymbol{\nu}_{\mathrm{ref}} }_{\mathrm{FF}} - \underbrace{ \mathbf{K} \hat{\boldsymbol{\epsilon}}_{\boldsymbol{y}} }_{\mathrm{FB}} - \underbrace{ \boldsymbol{\alpha}(\hat{\boldsymbol{x}}) }_{\mathrm{MC}} \biggr) \, ,
\end{align}
%
Fig.~\ref{fig:system} illustrates a nonlinear inversion-based tracking architecture \cite{grune2016nonlinear,Jing2020}, where the CE principle is implicitly invoked through exact state inversion (solid lines). In practice, the unobservability of $\boldsymbol{x}$ necessitates closing the loop using the state estimate $\hat{\boldsymbol{x}}$ (dashed path). Central to this study, this transition couples the observer dynamics into the tracking loop, inducing transients that disrupt the nominally decoupled error flow.

\begin{figure*}[t]
\centering
\begin{tikzpicture}[ 
  block/.style={draw, thick, minimum width=12mm, minimum height=1cm, fill=gray!6, rounded corners=1mm},
  block_small/.style={draw, thick, minimum width=12mm, minimum height=7mm, fill=gray!6, rounded corners=1mm},
  joint/.style={draw, circle, thick, minimum size=6.5mm, inner sep=0pt},
  arrow/.style={-{Stealth[length=6pt, width=6pt]}, rounded corners=1.5mm, line width=1.5pt},
  labelnode/.style={font=\footnotesize, inner sep=1pt},
  every node/.style={font=\footnotesize}
]

\node[block] (mod_ref) {\shortstack{Reference\vspace{.5mm}\\generator \eqref{eq:ref_diff}}};
\node[joint, right=13mm of mod_ref] (J_sub_1) {\large$\Delta$};
\node[block, right=15mm of J_sub_1] (mod_gain) {\shortstack{Control\vspace{.5mm}\\law \eqref{eq:error_dot_y}}};
\node[joint, right=12mm of mod_gain] (J_sum_1) {$\sum$};
\node[block, right=14mm of J_sum_1] (mod_di) {\shortstack{Dynamic\vspace{.5mm}\\inversion \eqref{eq:dynamic_inversion_1}}};
\node[block, right=11mm of mod_di] (plant) {\shortstack{Nominal\vspace{.5mm}\\plant \eqref{eq:state_space_GT}}};
\node[block, below right=6mm and 0mm of plant] (estimator) {\shortstack{State\vspace{.5mm}\\estimator \vspace{.5mm}\\ \eqref{eq:state_space_est}}};
\node[block, below=8mm of mod_di] (non_cancel) {\shortstack{Nonlinear \vspace{.5mm}\\ decoupling \eqref{eq:normal_form}}};
\node[block, below=9.7mm of J_sub_1] (Lie_alpha) {\shortstack{Output\vspace{.5mm}\\ observer \eqref{eq:normal_form}}};

\draw[arrow] ([xshift=-12mm]mod_ref.west) -- node[xshift=-1mm, above] {$\boldsymbol{y}_{\mathrm{ref}}$} (mod_ref);
\draw[arrow, orange!90!black] (mod_ref) -- node[xshift=2mm, above] {$\boldsymbol{y}_{\mathrm{ref}}^{(0:r-1)}$} (J_sub_1);
\draw[arrow, green!60!black] (J_sub_1) -- node[xshift=0mm, above] {$- \boldsymbol{\epsilon}_{ \boldsymbol{y}}^{(0:r-1)}$} (mod_gain);
\draw[arrow, green!60!black] (mod_gain) -- node[xshift=-2mm, above] {$- \dot{\boldsymbol{\epsilon}}_{\boldsymbol{y}}^\star$} (J_sum_1);
\draw[arrow, cyan!80!black](mod_ref) -- ++(0,0.9) -| node[xshift=-36.5mm, above]{$\boldsymbol{y}_{\mathrm{ref}}^{(r)}$}(J_sum_1);
\draw[arrow, magenta!90!white] (J_sum_1) -- node[xshift=0.8mm, above] {$\boldsymbol{\nu}$ \eqref{eq:virtual_inp}} (mod_di);
\draw[arrow] (mod_di) -- node[xshift=1.5mm, above] {$\boldsymbol{u}$} (plant);

\draw[arrow] (plant.east) -- ++(16mm,0) node[pos=0.35, above] {$\boldsymbol{y}$};
\draw[arrow](plant) -- ++(0,-33mm) -| node[xshift=100mm, yshift=6.3mm]{$\boldsymbol{x}$}(Lie_alpha.south);
\draw[arrow, orange!90!black] (Lie_alpha.north) -- node[xshift=-8mm] {$\boldsymbol{y}^{(0:r-1)}$} (J_sub_1.south);
\draw[arrow, orange!90!black] ([xshift=-4mm]non_cancel.north) -- node[left] {$\boldsymbol{\alpha}$} ([xshift=-4mm]mod_di.south);
\draw[arrow, orange!90!black] ([xshift=3mm]non_cancel.north) -- node[right] {$\boldsymbol{\beta}$} ([xshift=3mm]mod_di.south);
\draw[dashdotted, <-, >={Stealth[length=6pt, width=5pt]}, line width=1.5pt] (estimator.north) -- ++(0,11mm) node[above] {};
\draw[dashdotted, arrow] (estimator.south) -- ++(0,-9.2mm) -- ++(-14.5mm,0) node[xshift=12mm, yshift=6.2mm] {$\hat{\boldsymbol{x}}$};
\draw[<-, >={Stealth[length=6pt, width=5pt]}, line width=1.5pt] (non_cancel.south) -- ++(0,-9.5mm) node[above] {};
%
\draw[dashed, thick, rounded corners, fill=yellow!40!white, fill opacity=0.1, text opacity=1]($(mod_ref.north west)+(-0.37,1.05)$) rectangle node[xshift=-2.82cm, yshift=1.95cm]{Outer-loop (kinematics)}  ($(J_sum_1.south east)+(0.3,-2.6)$);
\draw[dashed, thick, rounded corners, fill=orange!40!white, fill opacity=0.1, text opacity=1]($(J_sum_1.north west)+(-0.44,1.325)$) rectangle node[xshift=0.8cm, yshift=1.95cm]{Inner-loop (dynamics)}  ($(non_cancel.south east)+(0.3,-0.498)$);
\end{tikzpicture}
\caption{Nonlinear tracking architecture: Nominally exact components (black) are shown alongside the linearized tracking law, comprising feedback (green), feedforward (blue), Lie-derivative synthesis (orange), and virtual input (magenta). The rightmost module (dashed) denotes the state estimator, departing from ideal CE assumptions to a practical feedback loop.}
\label{fig:system}
\end{figure*}

\subsection{Coupling Analysis}
The central mechanism of this study is the penetration of estimation-induced perturbations into the nominally decoupled control channels. We formalize this interaction using the generalized error decomposition
\begin{align} \label{eq:system_general}
\dot{\hat{\boldsymbol{\epsilon}}}_{\boldsymbol{x}} = \underbrace{ (\boldsymbol{F}_{\hat{\boldsymbol{x}}} - \boldsymbol{G}_{\hat{\boldsymbol{x}}} \mathbf{K}) \hat{\boldsymbol{\epsilon}}_{\boldsymbol{x}} }_{\text{\shortstack{Nominal CE\\dynamics}}} + \underbrace{\boldsymbol{\Omega}(\tilde{\boldsymbol{x}}, \boldsymbol{w}, \boldsymbol{v})}_{\text{\shortstack{Innovation\\coupling}}} + \underbrace{\boldsymbol{\delta}(\boldsymbol{x}, \hat{\boldsymbol{x}}, \boldsymbol{u})}_{\text{\shortstack{Structural\\residual}}} \, .
\end{align}
Here, $\boldsymbol{\Omega}(\cdot)$ represents the innovation coupling driving the parametric slaving in Prop.~\ref{prop:slave}, while $\boldsymbol{\delta}(\cdot)$ captures the geometric mismatch quantified as the input–output residual $\boldsymbol{\Delta}(\cdot)$ in Prop.~\ref{prop:res}.
\begin{proposition} \label{prop:slave}
The innovation coupling term $\boldsymbol{\Omega}(\cdot)$ acts as a bridge between estimation residuals and the controller. Consequently, tracking stability is not autonomous, but parametrically dependent on the estimation error $\tilde{\boldsymbol{x}}$ through the Jacobians \eqref{eq:jacobians}. This 'slaves' the control convergence to the geometric fidelity, forcing the observer to consume a portion of the system’s total stability margin.
\end{proposition}
The slaving described in Prop.~\ref{prop:slave} manifests as a failure of the algebraic cancellation in feedback linearization. The following propositions quantifies this departure.
\begin{proposition}[Residual Nonlinearity] \label{prop:res}
Consider the inversion law \eqref{eq:dynamic_inversion_2} implemented using the state estimate $\hat{\boldsymbol{x}}$. Under state uncertainty $\tilde{\boldsymbol{x}}$, the resulting input-output dynamics are perturbed as follows
\begin{align}
\boldsymbol{y}^{(r)} = \boldsymbol{\nu} + \boldsymbol{\Delta}(\boldsymbol{x}, \hat{\boldsymbol{x}}, \boldsymbol{u}) \, ,
\end{align}
where the output residual $\boldsymbol{\Delta}(\cdot)$ decomposes as
\begin{align}
\boldsymbol{\Delta}(\boldsymbol{x}, \hat{\boldsymbol{x}}, \boldsymbol{u}) \hspace{-0.5mm} = \hspace{-0.5mm} \bigl( \underbrace{ \boldsymbol{\alpha}(\boldsymbol{x}) - \boldsymbol{\alpha}(\hat{\boldsymbol{x}})  }_{\text{Drift mismatch}} \bigr) + \bigl( \underbrace{ \boldsymbol{\beta}(\boldsymbol{x}) -\boldsymbol{\beta}(\hat{\boldsymbol{x}}) }_{\text{Gain mismatch}} 
\hspace{-0.5mm} \boldsymbol{u} \bigr) .
\end{align}
Under the CE condition ($\hat{\boldsymbol{x}} \to \boldsymbol{x}$), this mismatch vanishes and exact feedforward is recovered. Otherwise, state uncertainty ($\tilde{\boldsymbol{x}} \neq \mathbf{0}$) acts as a non-negligible perturbation that disrupts the nominal integrator chain.
\end{proposition}
%
%
\begin{proposition} \label{prop:fallacy}
Consider an estimate-based feedback law $\boldsymbol{u} = \boldsymbol{\kappa}(\hat{\boldsymbol{x}})$. For nonlinear systems, the composition $\mathbf{g}(\boldsymbol{x})\boldsymbol{\kappa}(\hat{\boldsymbol{x}})$ defies linear decomposition into nominal and perturbation terms. Specifically, a first-order expansion yields
\begin{align}
\underbrace{ \mathbf{g}(\boldsymbol{x}) \boldsymbol{\kappa}(\hat{\boldsymbol{x}}) }_{\text{Actual } \boldsymbol{u}} \neq 
\underbrace{ \mathbf{g}(\boldsymbol{x}) \boldsymbol{\kappa}(\boldsymbol{x}) }_{\text{CE-based } \boldsymbol{u}} + \underbrace{ \nabla_{\boldsymbol{x}}(\mathbf{g}\boldsymbol{\kappa}) \tilde{\boldsymbol{x}} + \mathcal{O}( \| \tilde{\boldsymbol{x}} \|^2) }_{\text{Induced coupling}}\, ,
\end{align}
demonstrating the ``Fallacy'' that closed-loop dynamics can be partitioned into independent subsystems. 
\end{proposition}
Treating state acquisition as a purely deterministic process implicitly invokes a separation principle \cite{khalil1996robust}. While globally valid for LTI systems, it fails for high-order or non-Lipschitz plants. Consequently, the objective of this study is to characterize the stability of error manifold $\hat{\boldsymbol{\epsilon}}$ in regimes where the $\boldsymbol{\Omega}$ and $\boldsymbol{\delta}$ are non-negligible.

\section{Closed-Loop Stability Analysis} \label{sec:stability}
Leveraging the error decompositions in Props.~\ref{prop:slave}–\ref{prop:fallacy}, we employ Lyapunov theory to establish the convergence of the interconnected manifold, explicitly accounting for estimator-induced transients that perturb tracking.
\subsection{Lyapunov Stability}
To begin with, we demonstrate that despite strict violation of the separation principle, the system maintains practical stability within a calculable error envelope.
\begin{corollary} \label{Cor:lower_bound}
By the reverse triangle inequality, the perceived error magnitude satisfies
\begin{align} \label{eq:lower_bound}
\| \hat{\boldsymbol{\epsilon}}\| \triangleq \| \boldsymbol{\epsilon} - \tilde{\boldsymbol{x}} \| \geq \big| \| \boldsymbol{\epsilon} \| - \|\tilde{\boldsymbol{x}} \| \big| \, .
\end{align}
\end{corollary}
Corollary~\ref{Cor:lower_bound} reveals a crucial geometric constraint: even under ideal tracking ($\boldsymbol{\epsilon} \to \mathbf{0}$), the perceived error is lower-bounded by the residual, i.e., $\| \hat{\boldsymbol{\epsilon}} \| \to \| \tilde{\boldsymbol{x}} \|$. Analogous to the Cramér-Rao Lower Bound (CRLB), this result establishes an irreducible bias, dictated by the inherent uncertainty of the state estimate \cite{bendat2011random,ENGELSMAN2026112842}.
%
\begin{proposition} \label{prop:stability}
Consider the perceived error dynamics established in Sec.~\ref{sec:prob}. Let there exist a quadratic Lyapunov candidate \cite{lyapunov_2}
\begin{align} \label{eq:V_Lyp}
V(\hat{\boldsymbol{\epsilon}}) = \frac{1}{2}\hat{\boldsymbol{\epsilon}}^\top \boldsymbol{P} \hat{\boldsymbol{\epsilon}} = \frac{1}{2} \| \hat{\boldsymbol{\epsilon}} \|^2_{\boldsymbol{P}} \, , 
\end{align}
where $\boldsymbol{P} \succ 0$ solves the algebraic Lyapunov equation $\boldsymbol{F}^\top \boldsymbol{P} + \boldsymbol{P} \boldsymbol{F} = -\boldsymbol{Q}$ for a Hurwitz matrix $\boldsymbol{F}$ and $\boldsymbol{Q} \succ 0$. The time derivative along the closed-loop trajectories is
\begin{align} \label{eq:V_dot_neg}
\dot{V} \le \underbrace{ (-\lambda \| \hat{\boldsymbol{\epsilon}} \|^2 ) }_{\text{Dissipation}} + \underbrace{ (\mu \| \tilde{\boldsymbol{x}} \|^2 + d) }_{\text{Excitation}} \, ,
\end{align}
where $\lambda \triangleq \lambda_{\min}(\boldsymbol{Q})/\lambda_{\max}(\boldsymbol{P}) > 0$ is the nominal dissipation rate (convergence rate), $\mu > 0$ is the estimation coupling, and $d \ge 0$ is a uniform disturbance bound.
\end{proposition}
\begin{definition}
To characterize the coupling degree between the controller and the estimator, we define a dimensionless Separation Index (SI) as 
\begin{align}
\rho \triangleq \lambda / \mu \, .
\end{align}
\end{definition}

The inequality in \eqref{eq:V_dot_neg} captures the competition between dissipative feedback and estimation-induced excitation. The sign of $\dot{V}$, parametrically governed by $\rho$, yields three distinct qualitative regimes
\begin{enumerate}[label=\roman*), itemsep=1.5pt, parsep=1pt, topsep=2pt]
    \item \textbf{Contraction} ($\dot{V} < 0$): The controller efforts dominate the estimator coupling effects ($\rho \gg 1$), forcing error trajectories toward the invariant set.
    \item \textbf{Equilibrium} ($\dot{V} = 0)$: The controller exactly balances uncertainty effects ($\rho \approx 1$). This "stand-off" defines the performance ceiling where trajectories are trapped in a bounded oscillation.
    \item \textbf{Expansion} ($\dot{V} > 0$): The controller is overwhelmed by the uncertainty dominance ($\rho \ll 1$), potentially leading to trajectory divergence.
\end{enumerate}
This relationship formalizes the separation fallacy (cf. Prop. \ref{prop:fallacy}), demonstrating that the steady-state error floor is not an independent property of the observer, but is intrinsically governed by the structural ratio $\rho$.

\subsection{Invariant Error Tube}
The parametric balance established by the Lyapunov analysis determines the spatial bounds of the tracking error, which we now characterize in terms of geometric confinement and ultimate boundedness.

\begin{definition} \label{def:uub}
A trajectory $\hat{\boldsymbol{\epsilon}}(t)$ is uniformly ultimately bounded (UUB) with an ultimate bound $\gamma$ if there exists a $T(\hat{\boldsymbol{\epsilon}}(0), \gamma) \ge 0$ such that $\| \hat{\boldsymbol{\epsilon}} (t) \| \le \gamma$ for all $t \ge t_0 + T$ \cite{Barmish1983}.
\end{definition}
\begin{proposition} \label{Prop:uub}
Under bounded estimation error $\tilde{\boldsymbol{x}}$ and exogenous disturbances $d$, the perceived tracking error $\hat{\boldsymbol{\epsilon}}$ is UUB. Specifically, there exists an invariant set, the Error Tube $\mathcal{T}_\epsilon(t)$, with an outer radius defined by the dissipative ceiling
\begin{align} \label{eq:upper_bound}
\gamma(t) \triangleq \sqrt{ \frac{ \mu \| \tilde{\boldsymbol{x}}(t) \|^2 + d}{\lambda} } \, .
\end{align}
\end{proposition}
By reconciling the perception floor \eqref{eq:lower_bound} with the dissipative ceiling \eqref{eq:upper_bound}, we arrive at a unified characterization of the error's steady-state confinement.
\begin{corollary} \label{cor:tube}
Let Prop.~\ref{Prop:uub} and Cor.~\ref{Cor:lower_bound} hold. Then, $\hat{\boldsymbol{\epsilon}}$ is confined to the forward-invariant tube 
\begin{align} \label{eq:Inv_tube} 
\mathcal{T}_\epsilon (t) \triangleq \left\{ \, \hat{\boldsymbol{\epsilon}} \in \mathbb{R}^n : \| \tilde{\boldsymbol{x}}(t) \| \le \| \hat{\boldsymbol{\epsilon}} \| \le \gamma(t) \, \right\} \, .
\end{align}
Trajectories entering $\mathcal{T}_{\epsilon} (t)$ remain confined for all future time, establishing the practical closed-loop stability.
\end{corollary}
Fig.~\ref{fig:stab} illustrates the resulting geometry as a dynamic tension between control dissipation and perception uncertainty. While the ultimate bound $\gamma(t)$ fluctuates with the estimation error $\|\tilde{\boldsymbol{x}}(t)\|$, practical stability is preserved provided the dissipation rate is sufficiently dominant ($\rho \gg 1$). Structurally, the estimation residue $\|\tilde{\boldsymbol{x}}\|$ acts as a "noise floor"—the inner boundary of $\mathcal{T}_{\epsilon}$—that precludes asymptotic convergence to the physical origin ($\|\boldsymbol{\epsilon}\| \to 0$) regardless of the control gain.
\begin{figure}[h]
\centering
\begin{tikzpicture}[>=Latex,
node distance=16mm and 24mm,
block/.style = {draw, rounded corners=3pt, thick, align=center, inner sep=4pt, minimum width=14mm, minimum height=10mm},
labelnode/.style={font=\normalsize, inner sep=1pt},
every node/.style={font=\normalsize},  
line_ax/.style={-{Stealth[length=5pt, width=5pt]}, line width=1.65pt},
line_in/.style={gold!80!white, line width=1.9pt},
line_ot/.style={gold!70!green, line width=1.9pt},
d_arrow/.style={red!80!white, line width=2.1pt, dashed},
d_line/.style={line width=1.5pt, dashed},
frame/.style={line width=4pt, draw=gray!50},
circ/.style={thick, inner sep=0pt},
scale=\linewidth/0.8\textwidth] 
]
\def\Lx{0}      \def\Ly{0}
\def\Rx{8.0}    \def\Ry{-0.6}
\def\RoX{2.15}  \def\RoY{2.8}
\def\roX{0.9}  \def\roY{1.4}
\def\ROX{0.95}  \def\ROY{1.55}
\def\rOX{0.4}  \def\rOY{0.6}

\coordinate (O)     at (\Lx,\Ly);

\coordinate (LoTop) at (\Lx+\RoX, \Ly+\RoY);
\coordinate (LoBot) at (\Lx+\RoX, \Ly-\RoY);
\coordinate (RoTop) at (\Rx-\ROX, \Ry+\ROY);
\coordinate (RoBot) at (\Rx-\ROX, \Ry-\ROY);

\coordinate (LiTop) at (\Lx+\roX, \Ly+\roY);
\coordinate (LiBot) at (\Lx+\roX, \Ly-\roY);
\coordinate (RiTop) at (\Rx-\rOX, \Ry+\rOY);
\coordinate (RiBot) at (\Rx-\rOX, \Ry-\rOY);

\draw[line_ax] (O) -- ++(0,5);
\draw[line_ax] ($(O)+(3.75,2.5)$) -- ++(-7.5,-5);
\draw[line_ax] (O) -- ++(10.0,-0.8);

\def\RoT{-20}

\draw[line_ot] (\Lx,\Ly) ellipse [x radius=\RoX, y radius=\RoY, rotate=\RoT];
\draw[line_ot] (\Rx,\Ry) ellipse [x radius=\ROX, y radius=\ROY, rotate=\RoT];

\draw[line_ot]($(LoTop)+(-1.4,-0.1)$) .. controls ($(LoTop)+(3,-2)$) and ($(RoTop)+(2.,0.0)$) .. ($(RoTop)+(1.2,-0.05)$);
\draw[line_ot]($(LoBot)+(-2.25,0.1)$) .. controls ($(LoBot)+(1.1,0.8)$) and ($(LoBot)+(2,1.3)$) .. ($(RoBot)+(0.35,0.1)$);

\draw[line_in] (\Lx,\Ly) ellipse [x radius=\roX, y radius=\roY, rotate=\RoT];
\draw[line_in] (\Rx,\Ry) ellipse [x radius=\rOX, y radius=\rOY, rotate=\RoT];

\draw[line_in]($(LiTop)+(-.5,-0.05)$) .. controls ($(LiTop)+(1.1,-0.7)$) and ($(LiTop)+(3,-1.4)$) .. ($(RiTop)+(.5,0)$);
\draw[line_in]($(LiBot)+(-1.05,0.05)$) .. controls ($(LiBot)+(1,+0.4)$) and ($(LiBot)+(2.0,0.7)$) .. ($(RiBot)+(.2,0.02)$);
\draw[d_arrow]($(LoTop)+(-2.1,-0.2)$) .. controls ($(LoTop)+(1,-2.8)$) and ($(LoTop)+(4.0,-1.7)$) .. ($(RiTop)+(0.8,+0.4)$);

\node at ($(O)+(-2.9, -3.2)$) {$ \hat{{\epsilon}}_{\boldsymbol{\zeta}_1}(t) $};
\node at ($(O)+(-1.2, 4.5)$) {$ \hat{{\epsilon}}_{\boldsymbol{\zeta}_2}(t) $};
\node at ($(O)+(9.7, -0.2)$) {$t$};
\node[red!80!white, fill=white, text opacity=1, fill opacity=1.0, rotate=-5, inner sep=2.0pt] at (4.35,0.85) (cs_E){$\|  \widehat{\boldsymbol{\epsilon}}(t)\| $}; 
\node[gold!80!black, fill=white, text opacity=1, fill opacity=1.0, rotate=-5, inner sep=1.0pt] at (5.3,0.15) (cs_E){$\| \widetilde{\boldsymbol{x}}(t) \| $}; 
\node[green!40!orange, fill=white, text opacity=1, fill opacity=1.0, rotate=-10, inner sep=1.0pt] at (3.85,1.7) (cs_E){$\boldsymbol{\gamma}(t)$}; 

\end{tikzpicture}
\caption{Geometric representation of the stability region $\mathcal{T}_{\epsilon}$ in a generic phase plane ($\zeta_1, \zeta_2$). The perceived error $\| \hat{\boldsymbol{\epsilon}}(t) \|$ (red) is confined within an invariant annulus. The inner boundary (yellow) is defined by the estimation error $\|\tilde{\boldsymbol{x}} (t)\|$, while the outer boundary (green) represents the ultimate bound $\gamma  (t)$ induced by uncompensated disturbances.}
\label{fig:stab}
\end{figure}

\subsection{Topological Stability}
To conclude the stability analysis, Theorem~\ref{thm:thickness} reformulates the problem in geometric terms to establish a quantitative metric for the tracking envelope.
\begin{theorem} \label{thm:thickness}
Let the invariant tube thickness in \eqref{eq:Inv_tube} be defined as $\Delta \mathcal{T}_\epsilon(t) \triangleq \gamma(t) - \| \tilde{\boldsymbol{x}}(t) \|$. In the high-gain limit ($\rho \to \infty$), the tube thickness converges to the negative perception floor
\begin{align}
\lim_{\rho \to \infty} \Delta \mathcal{T}_\epsilon(t) = -\| \tilde{\boldsymbol{x}}(t) \| \, .
\end{align}
\end{theorem}
\begin{proof}
Factoring out $\| \tilde{\boldsymbol{x}}(t) \|$ and substituting the gain index $\rho = \lambda/\mu$ yields the explicit thickness expression
\begin{align}
\Delta \mathcal{T}_\epsilon(t) = \| \tilde{\boldsymbol{x}}(t) \| \left( \sqrt{ \rho^{-1} \left( 1 + \frac{d}{\mu \| \tilde{\boldsymbol{x}}(t) \|^2} \right) } - 1 \right) .
\end{align}
As dissipation dominates coupling ($\rho \to \infty$), the outer bound \eqref{eq:upper_bound} collapses ($\gamma \to 0$), and the thickness asymptotically approaches $-\| \tilde{\boldsymbol{x}}(t) \|$. 
\end{proof}

\section{Proposed Strategy} \label{sec:prop}
Analysis of CL sensitivity reveals that the CE idealism in Definition~\ref{def:CE} collapses under realistic conditions where $\boldsymbol{u} = \boldsymbol{\kappa}(\hat{\boldsymbol{\epsilon}}) \not\equiv \boldsymbol{\kappa}(\boldsymbol{\epsilon})$. Physically, this discrepancy implies that the controller operates on corrupted error coordinates, causing the state estimate to inadvertently alter the control topology.
To address this, we propose an \emph{estimation-aware} paradigm that regularizes the control law according to estimation quality and robustify the tracking manifold against state-space distortions.
\begin{table*}[b] 
\caption{Taxonomy of nonlinear control paradigms: Operational logic, nominal formulations, and our proposed EA counterparts.}
\centering
\footnotesize
\renewcommand{\arraystretch}{1.5} 
\begin{tabular}{|m{3.3cm}|m{5.1cm}|c|c|}
\hline
\textbf{Methods \& Works}\CC & \textbf{Principle}\CC & \textbf{Nominal Law ($\boldsymbol{\kappa}$)}\CC & \textbf{EA Law ($\boldsymbol{\kappa}_{\textit{EA}}$)}  \CC \\ \hline
{Cascaded Control (CC)} \cite{zhang2017cascaded,smeur2018cascaded,xu2020backstepping,CABUK2026111844} & Hierarchical decomposition into fast inner (attitude) and slow outer (position) loops. & 
$\begin{aligned} \boldsymbol{u} &= \boldsymbol{\pi}_{\text{in}}(\boldsymbol{\epsilon}_{\text{att}}, \boldsymbol{\nu}) \\ \boldsymbol{\nu} &= \boldsymbol{\pi}_{\text{out}}(\boldsymbol{\epsilon}_{\text{pos}}) \end{aligned}$ & 
$\begin{aligned} \boldsymbol{u}_{\textit{EA}} &= \boldsymbol{\pi}_{\text{in}}(\hat{\boldsymbol{\epsilon}}_{\text{att}}, \boldsymbol{\nu}, \eta) \\ \boldsymbol{\nu} &= \boldsymbol{\pi}_{\text{out}}(\hat{\boldsymbol{\epsilon}}_{\text{pos}}, \eta) \end{aligned}$ \\ \hline
{Incremental Nonlinear Dynamic Inversion (INDI)} \ \cite{smeur2016adaptive,di2016modeling,STEINERT2025103553} & Incremental control update using measured state-derivative increments ($\Delta \tilde{\boldsymbol{\chi}}_k$), canceling drift using sensors. & 
$\begin{aligned} \boldsymbol{u}_k &= \boldsymbol{u}_{k-1} + \mathbf{g}^{\dagger} \Delta \boldsymbol{\chi}_k \\ \Delta \boldsymbol{\chi}_k &\triangleq \boldsymbol{\nu}_k - \dot{\boldsymbol{x}}_k \end{aligned}$ & 
$\begin{aligned} \boldsymbol{u}_{\textit{EA},k} &= \boldsymbol{u}_{k-1} + \hat{\mathbf{g}}^{\dagger} \Delta \hat{\boldsymbol{\chi}}_k \\ \Delta \hat{\boldsymbol{\chi}}_k &\triangleq \boldsymbol{\nu}_k - \dot{\hat{\boldsymbol{x}}}_k + \boldsymbol{\Delta}(\eta_k) \end{aligned}$ \\ \hline
{Geometric Control (GC)} \cite{lee2010geometric,sreenath2013geometric,goodarzi2015geometric} & Coordinate-free control defined directly on the configuration manifold $SO(3)$ or $SE(3)$. & 
$\boldsymbol{u} = -\mathbf{K} \begin{pmatrix} \boldsymbol{\epsilon}_{\text{pos}} \\ \boldsymbol{\epsilon}_{\text{att}} \end{pmatrix} + \boldsymbol{u}_{\text{ff}}$ & 
$\boldsymbol{u}_{\textit{EA}} = -\mathbf{K}(\eta) \begin{pmatrix} \hat{\boldsymbol{\epsilon}}_{\text{pos}} \\ \hat{\boldsymbol{\epsilon}}_{\text{att}} \end{pmatrix} + \boldsymbol{u}_{\text{ff}}$ \\ \hline
{Differential Flatness (DF)} \cite{faessler2017differential,morrell2018differential,tal2020accurate} & Mapping of states and inputs to "flat" outputs via $\Phi(\cdot)$ and $\Psi(\cdot)$, to enable algebraic trajectory planning. & 
$\begin{aligned} \boldsymbol{x} &= \Phi(\boldsymbol{y}, \dots, \boldsymbol{y}^{(r)}) \\ \boldsymbol{u} &= \Psi(\boldsymbol{y}, \dots, \boldsymbol{y}^{(r)}) \end{aligned}$ & 
$\begin{aligned} \hat{\boldsymbol{x}} &= \Phi(\hat{\boldsymbol{y}}, \dots, \hat{\boldsymbol{y}}^{(r)}) \\ \boldsymbol{u}_{\textit{EA}} &= \Psi(\hat{\boldsymbol{y}}, \dots, \hat{\boldsymbol{y}}^{(r)}) + \boldsymbol{\Delta}(\eta) \end{aligned}$ \\ \hline
{Feedback Linearization (FL) / Dynamic Inversion (DI)} \cite{voos2009nonlinear,das2009dynamic,freddi2011feedback} & Algebraic inversion of the input–output map to cancel nonlinearities and enforce linear dynamics. & 
$\begin{aligned} \boldsymbol{u} &= \boldsymbol{\alpha}(\boldsymbol{x}) + \boldsymbol{\beta}(\boldsymbol{x}) \boldsymbol{\nu} \\ \boldsymbol{\alpha} &\triangleq -\mathbf{g}^{\dagger} \boldsymbol{f} \\ \boldsymbol{\beta} &\triangleq \mathbf{g}^{\dagger} \end{aligned}$ & 
$\begin{aligned} \boldsymbol{u}_{\textit{EA}} &= \hat{\boldsymbol{\alpha}} + \hat{\boldsymbol{\beta}} \boldsymbol{\nu} + \boldsymbol{\Delta}(\eta) \\ \hat{\boldsymbol{\alpha}} &\triangleq \boldsymbol{\alpha}(\hat{\boldsymbol{x}}) \\ \hat{\boldsymbol{\beta}} &\triangleq \boldsymbol{\beta}(\hat{\boldsymbol{x}}) \end{aligned}$ \\ \hline
{Nonlinear MPC (NMPC)} \cite{kang2009linear,eskandarpour2020constrained,wang2021efficient,LIU2026111657} & Receding-horizon cost minimization subject to nonlinear dynamics $\dot{\boldsymbol{x}}$ and constrained state $\mathcal{X}$ and input $\mathcal{U}$ sets. & 
$\begin{aligned} \boldsymbol{u}^* = \arg \min_{\boldsymbol{u} \in \mathcal{U}} \int_{\Delta t} \ell_c (\cdot) d\tau \\ \text{s.t.}  \  \dot{\boldsymbol{x}} = \boldsymbol{f}(\boldsymbol{x}, \boldsymbol{u}), \, \boldsymbol{x} \in \mathcal{X} \end{aligned}$ & 
$\begin{aligned} \boldsymbol{u}_{\textit{EA}}^* = \arg \min_{\boldsymbol{u} \in \mathcal{U}} \int_{\Delta t} \ell_c (\cdot, \eta) d\tau \\ \text{s.t.} \ \dot{\hat{\boldsymbol{x}}} = \boldsymbol{f}(\hat{\boldsymbol{x}}, \boldsymbol{u}), \, \boldsymbol{x} \in \mathcal{X}(\eta) \end{aligned}$ \\ \hline
\end{tabular}
\label{t:paradigms}
\end{table*}

\begin{definition} \label{def:aware}
We define a control law to be Estimation-Aware (EA) if it explicitly incorporates a measurable descriptor of uncertainty, $\eta(t)$, into the feedback structure. Formally, an EA control law $\boldsymbol{\kappa}_{\textit{EA}}(\cdot)$ is implemented as
\begin{align} \label{eq:EA_law}
\boldsymbol{u}_{\textit{EA}} = \boldsymbol{\kappa}_{\textit{EA}} \big( \hat{\boldsymbol{\epsilon}}, \eta(t) \big) \, ,
\end{align}
where $\eta(t) \in \mathbb{R}_{\ge 0}$ denotes a scalar metric of uncertainty (e.g., $\|\tilde{\boldsymbol{x}}(t)\|$, $\mathrm{tr}(\mathbf{P}(t))$, or the tube radius $\gamma(t)$). The resulting closed-loop system is said to be EA-stable if it satisfies the practical tracking bound
\begin{align}
\limsup_{t \to \infty} \| \boldsymbol{\epsilon}(t) \| \leq \xi \big( \eta(t) \big) \, ,
\end{align}
where $\xi(\cdot)$ is a class-$\mathcal{K}$ function mapping the estimation uncertainty to an ultimate tracking error bound \cite{Barmish1983,grune2016nonlinear}.
\end{definition}
As a result, stability shifts from asymptotic convergence to forward invariance within the tube $\mathcal{T}_{\epsilon} (t)$, in line with the bounds derived in Section~\ref{sec:stability} and illustrated in Fig.~\ref{fig:stab}.

Implementing the design philosophy established in Definition~\ref{def:aware} requires identifying specific "junctions" within the control architecture where the uncertainty metric $\eta(t)$ can be integrated. 
Based on the structural properties of the nonlinear paradigms surveyed in this work, we categorize the EA implementation into three fundamental functional domains described below:
\\
\textbf{1. Additive Awareness (Signal-Level):} This layer injects a robustifying signal $\boldsymbol{\Delta}(\eta)$ to counteract residual nonlinearities arising from imperfect cancellation. By acknowledging that state uncertainty biases the nominal control effort, $\boldsymbol{\Delta}(\eta)$ serves as a dynamic compensator that reconciles the divergence between perceived and actual system dynamics.
\\
\textbf{2. Parametric Awareness (Sensitivity-Level):} Here, the uncertainty metric $\eta(t)$ modulates the controller bandwidth via a gain-governing matrix $\mathbf{K}(\eta)$. This approach regulates closed-loop sensitivity by "backing off" gains when estimation confidence is low, thereby preventing the amplification of estimation noise and high-frequency residuals into the actuators.
\\
\textbf{3. Geometric Awareness (Manifold-Level):} In optimization paradigms, awareness is embedded into the topology of the problem. The metric $\eta(t)$ dictates the admissible volume of the safe operating region by dynamically contracting the constraint set $\mathcal{X}$ or regularizing the stage cost $\ell_c$. This mechanism effectively "shrinks" the admissible error manifold as uncertainty grows, ensuring set-invariance for the true state.
\\
In summary, Table~\ref{t:paradigms} maps contemporary control formulations ($\boldsymbol{\kappa}$) to their estimation-aware variants ($\boldsymbol{\kappa}_{\textit{EA}}$) for high-bandwidth aerial platforms \cite{kendoul2012survey,engelsman2025guidance}. 

\subsection{Architecture and Mechanism} \label{remark:choice}
While the INDI formulation inherently provides robust tracking against plant uncertainties, sensor noise inevitably distorts the state-observer $\widehat{\mathcal{S}}$. Given the central role of estimation-control coupling explored in this study, this distortion severely degrades CL stability margins. To counter this operational blindness, the proposed EA augmentation $\boldsymbol{\Delta}(\boldsymbol{\eta}_k)$ is embedded directly within the nominal INDI control derivative
\begin{align} \label{eq:u_EA}
\dot{\boldsymbol{u}}_{\textit{EA}}(t) &= \hat{\mathbf{g}}^{\dagger} \dot{\hat{\boldsymbol{\chi}}}(t) \, ,
\\
\dot{\hat{\boldsymbol{\chi}}}(t) &\triangleq \boldsymbol{\nu}(t) - \dot{\hat{\boldsymbol{x}}}(t) + \boldsymbol{\Delta}(\boldsymbol{\eta}(t)) \, .
\end{align}
This awareness is formalized by converting the posterior error covariance matrix $\mathbf{P}(t)$ into an uncertainty-gated operator that selectively cross-projects $\dot{\hat{\boldsymbol{x}}}(t)$ via
\begin{align} \label{eq:EA_act}
\boldsymbol{\Delta}(\boldsymbol{\eta}) &= \left( \mathbf{I} - \overline{\mathbf{P}} \right) \dot{\hat{\boldsymbol{x}}} , \, 
\\
\overline{\mathbf{P}} &\triangleq \boldsymbol{D}^{-1/2} \mathbf{P} , \boldsymbol{D}^{-1/2} \, ,
\end{align}
where $\boldsymbol{D} = \mathrm{diag}(\mathbf{P})$ and $\overline{\mathbf{P}}$ denotes the normalized correlation matrix (with time dependencies suppressed for brevity). In low-velocity flight, estimation certainty remains high and state uncertainties evolve independently; thus, $\overline{\mathbf{P}} \to \mathbf{I}$, causing the augmentation vector $\boldsymbol{\Delta}(\boldsymbol{\eta})$ to vanish and seamlessly recovering the nominal INDI control law. Conversely, during high-velocity maneuvers or severe perception degradation, the off-diagonal coupling terms in $(\mathbf{I} - \overline{\mathbf{P}})$ expand. The operator then cross-projects these correlated state rates, acting as an algorithmic buffer that dampens $\dot{\hat{\boldsymbol{\chi}}}$ to suppress parasitic actuation on distorted state estimates.

\subsection{Quadrotor Physical Realization}
The system dynamics are mapped to the physical quadrotor realization illustrated in Fig.~\ref{fig:quad_manu}. We adopt a standard six-degree-of-freedom (6-DoF) rigid-body model characterized by mass $m$ and moment of inertia (MoI) matrix $\boldsymbol{J}$. To avoid representation singularities, the attitude is parameterized via a unit quaternion
\begin{align} \label{eq:quat_def}
\boldsymbol{q} \triangleq \boldsymbol{q}_{\mathcal{B}}^{\mathcal{I}}
=
\begin{pmatrix}
q_0 & q_1 & q_2 & q_3
\end{pmatrix}^\top \ , \ \|\boldsymbol{q}\|=1 \, ,
\end{align}
which is mapped to the standard Euler angle orientation vector $\boldsymbol{\Theta}$, comprising roll ($\phi$), pitch ($\theta$), and yaw ($\psi$), defined via a $Z$-$Y$-$X$ rotation sequence. The corresponding rotation matrix $\boldsymbol{R}(\boldsymbol{q}) \in SO(3)$ maps vectors from the body frame $\mathcal{B}$ to the inertial frame $\mathcal{I}$, and is defined by its columns as
\begin{align}
\boldsymbol{R}(\boldsymbol{q}) \triangleq \begin{bmatrix} \mathbf{e}_x^\mathcal{B} & \mathbf{e}_y^\mathcal{B} & \mathbf{e}_z^\mathcal{B} \end{bmatrix}_{\mathcal{I}} ,
\end{align}
where $\mathbf{e}_i^\mathcal{B}$ denote the principal axes of $\mathcal{B}$ expressed in $\mathcal{I}$. Let $\boldsymbol{\xi}^\mathcal{I}=(x,y,z)^\top$ denote inertial position and $\boldsymbol{v}^\mathcal{B}$ the body linear velocity. The translational kinematics are
\begin{align}
\dot{\boldsymbol{\xi}}^\mathcal{I}
=
\boldsymbol{R}(\boldsymbol{q})\,\boldsymbol{v}^\mathcal{B},
\label{eq:trans_kin}
\end{align}
while the body angular velocity $\boldsymbol{\omega}^\mathcal{B}$ drives the quaternion evolution according to \cite{carino2015quadrotor}
\begin{align} \label{eq:quat_kinematics}
\dot{\boldsymbol{q}}
=
\frac{1}{2}
\boldsymbol{q}\odot
\begin{pmatrix}
0 \\
\boldsymbol{\omega}^\mathcal{B}
\end{pmatrix},
\end{align}
where $\odot$ denotes Hamilton quaternion multiplication. 
Compiling these definitions, the state vector is given as
\begin{align} \label{eq:states}
\boldsymbol{x} = \begin{pmatrix}
(\boldsymbol{\xi}^\mathcal{I})^\top &
(\boldsymbol{v}^\mathcal{B})^\top &
\boldsymbol{q}^\top &
(\boldsymbol{\omega}^\mathcal{B})^\top
\end{pmatrix}^\top
\in \mathbb{R}^{13} \, .
\end{align}
The control input vector $\boldsymbol{u}$ parameterizes the motor-induced vertical thrust $\boldsymbol{T}$ and body torques $\boldsymbol{\tau}$ as
\begin{align} \label{eq:u_com}
\boldsymbol{u} \triangleq \bigg( u_z^\mathcal{B} \ 
\underbrace{
\left( {\tau}_{x}^\mathcal{B} \ {\tau}_{y}^\mathcal{B} \ {\tau}_{z}^\mathcal{B} \right) }_{ \boldsymbol{u}_{\tau}^\mathcal{B} } \bigg)^\top 
\in \, \mathbb{R}^4 \, .
\end{align}
Accounting for external aerodynamic forces $\boldsymbol{f}_{ \mathrm{aero} }^\mathcal{B}$ and moments $\boldsymbol{\tau}_{\mathrm{aero}}^\mathcal{B}$, the total applied body-frame wrench is
\begin{align} \label{eq:wrench}
\begin{pmatrix}
\boldsymbol{f}^\mathcal{B} \\
\boldsymbol{\tau}^\mathcal{B}
\end{pmatrix} = \begin{pmatrix}
u_z^\mathcal{B} \mathbf{e}_z + \boldsymbol{f}_{\mathrm{aero}}^\mathcal{B} \\
\boldsymbol{u}_{\tau}^\mathcal{B} + \boldsymbol{\tau}_{\mathrm{aero}}^\mathcal{B}
\end{pmatrix} \, .
\end{align}
Substituting these expressions into the rigid-body equations yields the full control-affine state-space dynamics
\begin{align} \label{eq:x_dot_big}
\dot{\boldsymbol{x}} = \underbrace{ \begin{pmatrix}
\boldsymbol{R} (\boldsymbol{q}) \boldsymbol{v}^\mathcal{B} \\
\boldsymbol{R}^\top(\boldsymbol{q}) \textbf{g}^\mathcal{I} - \boldsymbol{\omega}^\mathcal{B} \times \boldsymbol{v}^\mathcal{B} + \frac{1}{m} \boldsymbol{f}^\mathcal{B} \\
\frac{1}{2}\boldsymbol{q} \odot
\begin{pmatrix}
0 \\ \boldsymbol{\omega}^\mathcal{B}
\end{pmatrix} \\
\boldsymbol{J}^{-1}
\left( - \boldsymbol{\omega}^\mathcal{B}
\times \boldsymbol{J}\boldsymbol{\omega}^\mathcal{B} + \boldsymbol{\tau}^\mathcal{B}
\right)
\end{pmatrix} 
}_{\boldsymbol{f}(\boldsymbol{x}) + \mathbf{g}(\boldsymbol{x})\boldsymbol{u} + \boldsymbol{w}} \, .
\end{align}
Detailed aerodynamic derivations and physical parameters are given in \cite{engelsman2026czupt} and \cite{yu2020two}, respectively.
\begin{figure}[!h]
\centering
\begin{tikzpicture}[>=Latex,
node distance=16mm and 24mm,
block/.style = {draw, rounded corners=3pt, thick, align=center, inner sep=4pt, minimum width=14mm, minimum height=10mm},
labelnode/.style={font=\normalsize, inner sep=1pt},
every node/.style={font=\normalsize},  
arrow/.style={-{Stealth[length=5pt, width=5pt]}, line width=1.95pt},
d_arrow/.style={-{Stealth[length=5pt, width=5pt]}, line width=1.6pt, dashed},
d_line/.style={line width=1.5pt, dashed},
frame/.style={line width=4pt, draw=gray!50},
circ/.style={thick, inner sep=0pt},
]
\node[inner sep=0pt, align=center, opacity=0.6] (quad) {\includegraphics[width=0.45\textwidth]{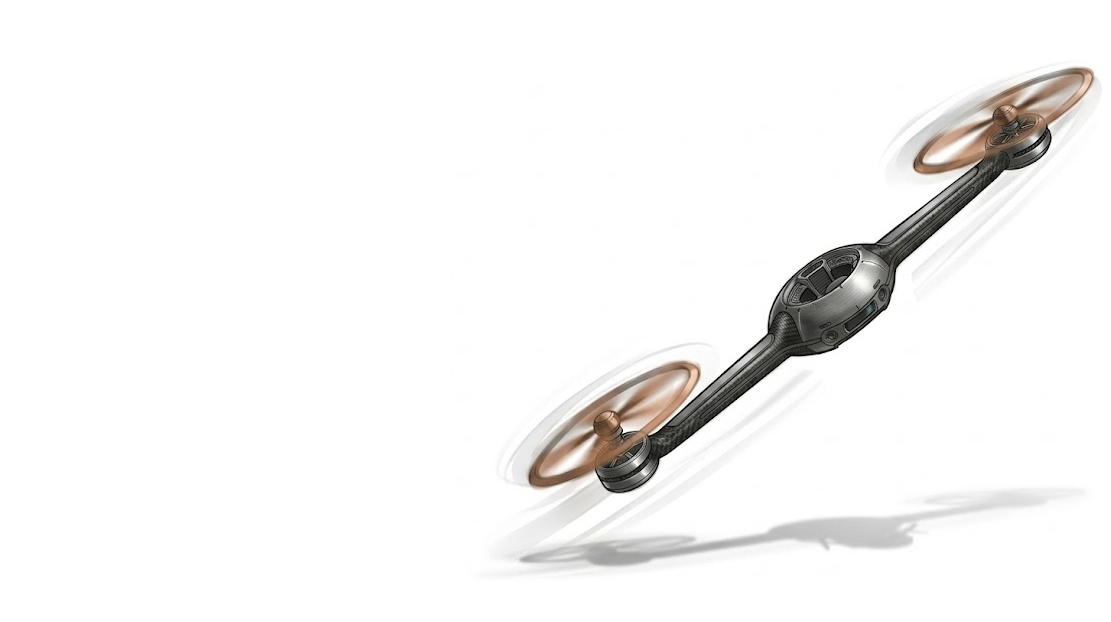}};

\node[circ, shift={(+20mm,+15mm)}] at (-6,-2.7) (cs_I){};
\draw[arrow] (cs_I) --++ (1.25*.8,-.5*.8) node[shift={(3mm,-1mm)}]{$\textbf{e}_x^\mathcal{I}$};
\draw[arrow] (cs_I) --++ (1.25*.8,.35*.8) node[shift={(3mm,0mm)}]{$\textbf{e}_y^\mathcal{I}$};
\draw[arrow] ($(cs_I)+(0.02,-0.02)$) --++ (0.035*.8, 1.2*.8) node[above]{$\textbf{e}_z^\mathcal{I}$};

\node[circ, shift={(+0mm,-2.5mm)}] at (2.05,0.27) (cs_B){};
\draw[d_arrow] ($(cs_B)+(-.05,0.05)$) --++ (-0.66,0.77) node[shift={(-1.8mm,1.0mm)}, fill=white, text opacity=1, fill opacity=0.3]{${\boldsymbol{T}}$};
\draw[arrow] (cs_B) --++ (-1.3,-0.2) node[shift={(-2.4mm,0.2mm)}]{$\textbf{e}_y^\mathcal{B}$};
\draw[arrow] (cs_B) --++ (-0.5,-0.95) node[shift={(-0.5mm,-2.9mm)}]{$\textbf{e}_x^\mathcal{B}$};
\draw[arrow] (cs_B) --++ (+.7,-.87) node[shift={(2mm,-2.mm)}]{$\textbf{e}_z^\mathcal{B}$};
\draw[d_arrow] (cs_B) --++ (0,-1.5) node[below]{$m \textbf{{g}}^\mathcal{I}$};
%
\node[circ, red!80!white, shift={(+2mm,+1mm)}, rotate=-5] at (.44,0.35) (cs_E){$\ \hat{\boldsymbol{\epsilon}} \ $}; 
\draw[arrow, red!80!white]($(cs_E)+(0.2,-0.07)$) --++ (0.72, -0.18) node[above]{};
\draw[arrow, red!80!white]($(cs_E)+(-0.2,0.02)$) --++ (-0.75, .175) node[above]{};
\node at ($(cs_I)+(0,-.5)$) [fill=white, text opacity=1, fill opacity=0.0, shift={(48mm,34.5mm)}]{$\hat{\boldsymbol{x}}$};
\node at ($(cs_I)+(0,-.5)$) [fill=white, text opacity=1, fill opacity=0.0, shift={(30mm,29mm)}]{$\boldsymbol{x}_{\mathrm{ref}} $};

\draw[blue!20!green, line width=3.0pt, dashdotdotted]($(-4.2,1.6)$) .. controls ($(-1.5,1.9)$) and  ($(1.9,1.8)$) .. ($(2.1,0.45)$);
\draw[black!70!white, line width=3.0pt, dashdotdotted]($(-4.2,1.5)$) .. controls ($(-2.4,1.7)$) and  ($(1.5,1.6)$) .. ($(-0.45,0.7)$);
\end{tikzpicture}
\caption{Transient tracking geometry: The perceived tracking error (red) illustrates the spatial deviation of the estimated trajectory (green) from the target reference (black).}
\label{fig:quad_manu}
\end{figure}

\subsection{Frozen-Time Snapshot} 
To leverage linear control synthesis techniques, the nonlinear dynamics in \eqref{eq:x_dot_big} are analyzed locally about a quasi-static operating point $(\boldsymbol{x}_0,\boldsymbol{u}_0)$, yielding
\begin{align}
\widehat{\mathcal{S}}_0 : \left\{ 
\begin{aligned}
\dot{\tilde{\boldsymbol{x}}} &= \boldsymbol{F}_{\boldsymbol{x}_0} \tilde{\boldsymbol{x}} + \boldsymbol{G}_{\boldsymbol{x}_0} \tilde{\boldsymbol{u}} + \boldsymbol{w} \, ,
\\
\tilde{\boldsymbol{y}} &= \boldsymbol{H}_{\boldsymbol{x}_0} \tilde{\boldsymbol{x}} + \boldsymbol{v} \, .
\end{aligned} \right.
\end{align}
Unlike the time-varying estimation variant in \eqref{eq:error_space}, the frozen-time paradigm \cite{Chul2004frozen} assumes that, even during aggressive maneuvers, the local Jacobians vary slowly relative to the control horizon ($\dot{\boldsymbol{F}} \approx \mathbf{0}, \dot{\boldsymbol{G}} \approx \mathbf{0}$). This time-scale separation permits treating the system as locally LTI, capturing a momentary operating snapshot suitable for localized analysis. Consequently, applying the Laplace transform under zero initial conditions yields the nominal open-loop MIMO transfer matrix
\begin{align} \label{eq:OL_TF}
\boldsymbol{\Gamma}(s) =
\boldsymbol{H}_{\boldsymbol{x}_0}
\left( s\boldsymbol{I} - \boldsymbol{F}_{\boldsymbol{x}_0} \right)^{-1} \boldsymbol{G}_{\boldsymbol{x}_0} \ \in \ \mathbb{C}^{p \times m} \, .
\end{align}
The zero-pole-gain form of the $ij$-th entry is then
\begin{align} \label{eq:pole_zero}
\Gamma_{ij}(s) =  K_{zp,ij} \frac{\prod_{m=1}^{N_{z}}(s-z_{m})}
{\prod_{k=1}^{N_{p}}(s-p_{k})} \, ,
\end{align}
where $K_{{zp},ij}$ denote the channel gain, $z_{m}$ and $p_{k}$ denote the respective zeros and poles, and $N_{z}$ with $N_{p}$ their corresponding total counts. To explicitly isolate tracking bandwidth and high-frequency roll-off characteristics, the associated time-constant Bode form is 
\begin{align} \label{eq:bode}
\Gamma_{ij}(s) = \frac{\kappa_{ij}}{s^\nu} \frac{\prod_{m=1}^{N_z} (1 + \tau_m s)}{\prod_{k=1}^{N_p} (1 + \tau_k s)} \prod_{l=1}^{N_c} \frac{1}{\frac{s^2}{\omega_l^2} + \frac{2\zeta_l}{\omega_l}s + 1} \, ,
\end{align}
where $\kappa_{ij}$ represents the static plant gain of the $ij$-th entry, $\nu$ is the type index, $\tau$ denotes a first-order time constant, and $N_{c}$ represent the counts of underdamped second-order complex conjugate pairs.
\\
\textit{Dominant Channels:} The multivariable MIMO matrix maps an extensive array of coupled transient modes, many of which remain negligible to primary flight characteristics. To maintain tractability, the analysis is restricted to the dominant diagonal channels of the system, namely the SISO transfer functions $\Gamma_{z}^{T}$ (heave), $\Gamma_{\phi}^{\tau_x}$ (roll), $\Gamma_{\theta}^{\tau_y}$ (pitch), and $\Gamma_{\psi}^{\tau_z}$ (yaw).
These frequency-domain representations enable benchmarking via two distinct structural perspectives:
\\
\textit{1. Open-Loop Analysis:} The gain structure in \eqref{eq:pole_zero} tracks pole-zero migrations across the flight envelope, whereas the Bode form in \eqref{eq:bode} isolates the tracking bandwidth, transmission zeros, and stability margins.
\\
\textit{2. Closed-Loop (CL) Analysis:} Introducing a static feedback gain snapshot $\boldsymbol{K} \in \mathbb{R}^{m \times p}$ yields the localized control law $\tilde{\boldsymbol{u}} = -\boldsymbol{K}\tilde{\boldsymbol{y}}$. Closing the loop establishes the autonomous error dynamics
\begin{align} \label{eq:F_CL}
\dot{\tilde{\boldsymbol{x}}} = \underbrace{(\boldsymbol{F}_{\boldsymbol{x}_0} - \boldsymbol{G}_{\boldsymbol{x}_0} \boldsymbol{K}\boldsymbol{H}_{\boldsymbol{x}_0})}_{\boldsymbol{F}_{\mathrm{CL}}} \tilde{\boldsymbol{x}} \, .
\end{align}
Via the open-loop transfer matrix in \eqref{eq:OL_TF}, the resulting CL frequency-domain behavior is characterized by the complementary sensitivity matrix
\begin{align} \label{eq:sensitivity}
\overline{\boldsymbol{\Gamma}}_{\mathrm{CL}} (s) = \left( \boldsymbol{I}_p + \boldsymbol{\Gamma}\boldsymbol{K} \right)^{-1} \boldsymbol{\Gamma}\boldsymbol{K} \in \mathbb{C}^{p \times p} . 
\end{align}
%

\section{Numerical Simulation} \label{sec:results}
This section stress-tests the proposed EA augmentation under severe model uncertainty and sensor noise. 

\subsection{Stability Analysis}
Grounded in Lyapunov stability principles, we first provide a geometric interpretation of the bounds established in Theorem~\ref{thm:thickness}. To achieve this, we adopt a phase-portrait representation of the radial error projection.
\begin{align}
\dot{\hat{\boldsymbol{\epsilon}}}_{\parallel} \le - \lambda \left( 1 - \frac{\gamma^2}{ \| \hat{\boldsymbol{\epsilon}} \|^2 } \right) \hat{\boldsymbol{\epsilon}} \, .
\end{align}
As evidenced by this inequality, stabilizability is dictated largely by the trade-off between the convergence rate $\lambda$ and the error bound $\gamma$, which naturally partitions the asymptotic analysis into four distinct regimes:
\begin{enumerate}[label=\roman*), itemsep=1.5pt, parsep=1pt, topsep=2pt]
\item \textbf{Stable Node ($ \gamma \ll \|\hat{\boldsymbol{\epsilon}}\| $):} The perception floor vanishes, and the dynamics simplify to $\dot{\hat{\boldsymbol{\epsilon}}} \approx -\lambda \hat{\boldsymbol{\epsilon}}$. 
%
\item \textbf{Unstable Source ($ \gamma > \| \hat{\boldsymbol{\epsilon}} \| $):} The bracketed term flips sign, converting the origin into an unstable source that drives tracking errors outward $\dot{\hat{ \boldsymbol{\epsilon}}} \approx \lambda \hat{\boldsymbol{\epsilon}}$. 
%
\item \textbf{Stable Limit-Cycle ($ \gamma = \| \hat{\boldsymbol{\epsilon}} \| $):} The radial velocity nullifies ($\dot{\hat{\boldsymbol{\epsilon}}} = \boldsymbol{0}$), defining an isolated periodic orbit whose tracking precision is bounded by $\gamma$.
%
\item \textbf{Singularity ($ \lim_{\lambda \to \infty} \gamma = 0 $):} Exploding gains drive $\gamma$ to zero, revealing that even infinite control effort fails to improve tracking precision as $\| \tilde{\boldsymbol{x}} \| \leq \| \hat{\boldsymbol{\epsilon}} \|$.
\end{enumerate}
%
As visualized in Fig.~\ref{fig:phase}, driving $\hat{\boldsymbol{\epsilon}} \to \boldsymbol{0}$ fails to eliminate the physical tracking error, which remains strictly lower-bounded by the estimation error floor $\| \tilde{\boldsymbol{x}} \|$. As the trajectory collapses onto this irreducible boundary, the separation fallacy formalized in Prop.~\ref{prop:fallacy} is exposed.
\begin{figure}[!h]
\begin{center}
\includegraphics[width=0.45\textwidth]{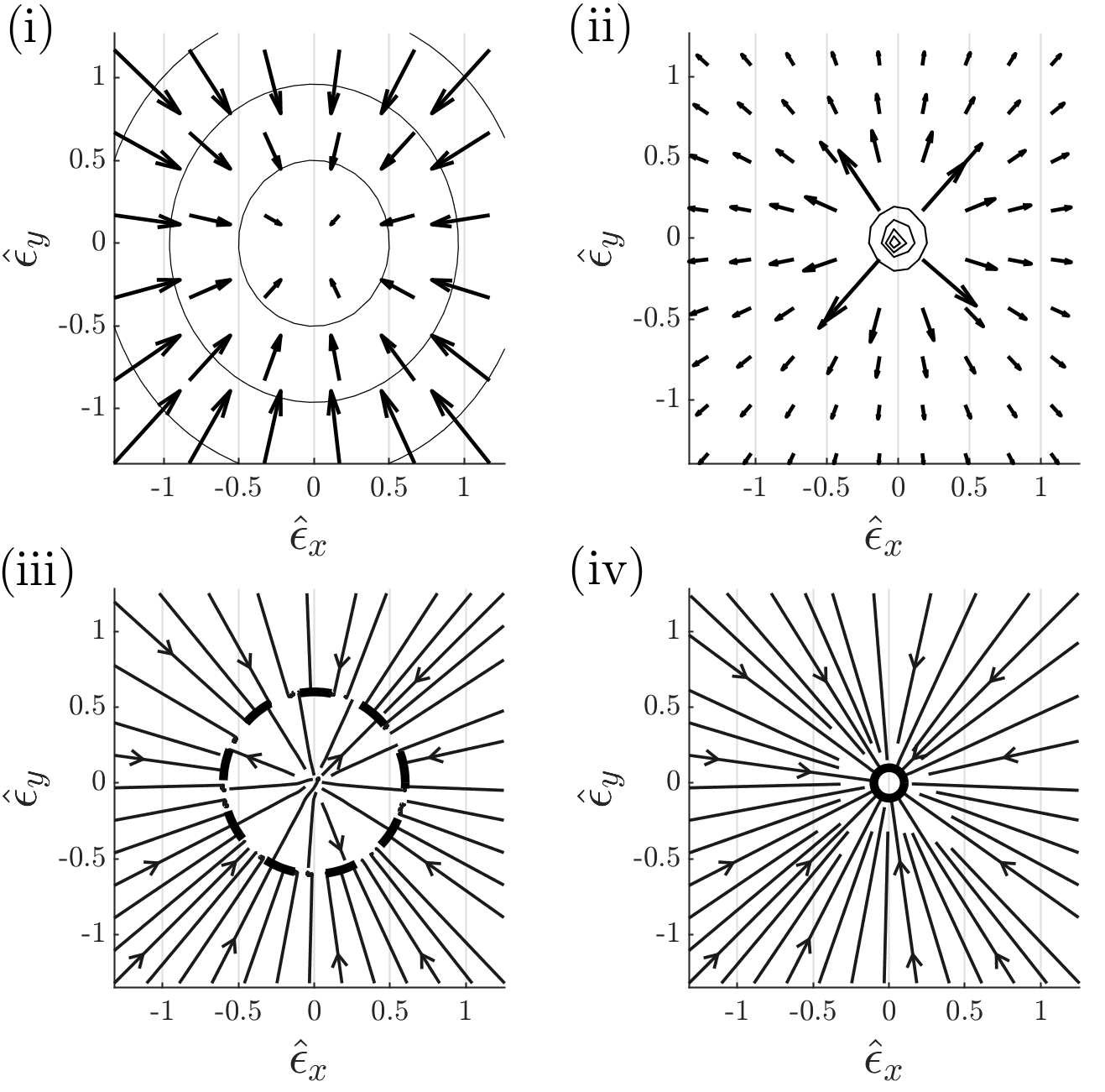}
\caption{Four-regime phase-portrait analysis: i) stability, ii) instability, iii) equilibrium, and iv) singularity.}
\label{fig:phase}
\end{center}
\end{figure} 

\subsection{Baseline Analysis}
A tilted Lemniscate trajectory is selected to evaluate performance along spatial inversions that excite the vehicle's dynamic modes. 
We establish a clear baseline for subsequent comparisons using two pillars. First, the control logic is dictated by the INDI framework presented in Table~\ref{t:paradigms}, applied to both nominal $\boldsymbol{u}$ and EA-augmented $\boldsymbol{u}_{\textit{EA}}$ controllers. Second, a tilted Lemniscate is selected as the 3D reference trajectory, creating a figure-eight path enclosed within a $2\mathcal{R} \times \mathcal{R} \times 0.5 \mathcal{R}$ bounding box with $\mathcal{R}=5$m. 
We begin our analysis in a benign, quasi-static aerial operating regime featuring a low tracking velocity of $\mathrm{v}_{\mathrm{ref}} = 1.0\,\text{m/s}$, which corresponds to a full orbital period of $\text{T} = 46.3\,\text{s}$. Under the baseline controller $\boldsymbol{u}$, Fig.~\ref{fig:lemni_1} displays the spatial coordinates $\boldsymbol{\xi}^{\mathcal{I}}(t)$ of the ground truth (blue), the estimates (orange), and the reference path (dashed black).
\begin{figure}[!h]
\begin{center}
\includegraphics[width=0.47\textwidth]{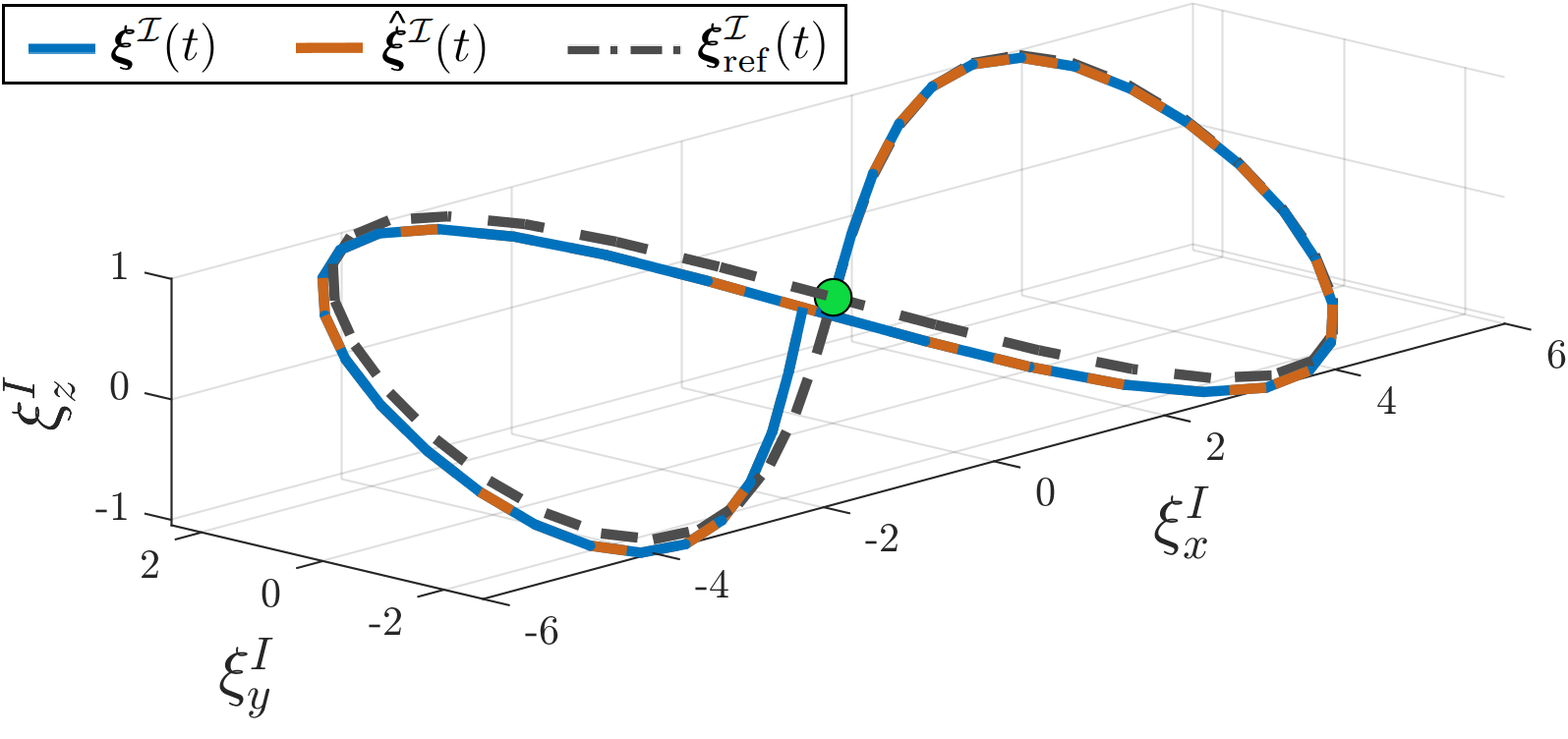}
\caption{Tracking scenario at $\mathrm{v}_{\mathrm{ref}} = 1.0\,\text{m/s}$: spatial coordinates of three kinematic trajectories under baseline controller $\boldsymbol{u}$.}
\label{fig:lemni_1}
\end{center}
\end{figure} 
\begin{figure}[!t]
\begin{center}
\includegraphics[width=0.485\textwidth]{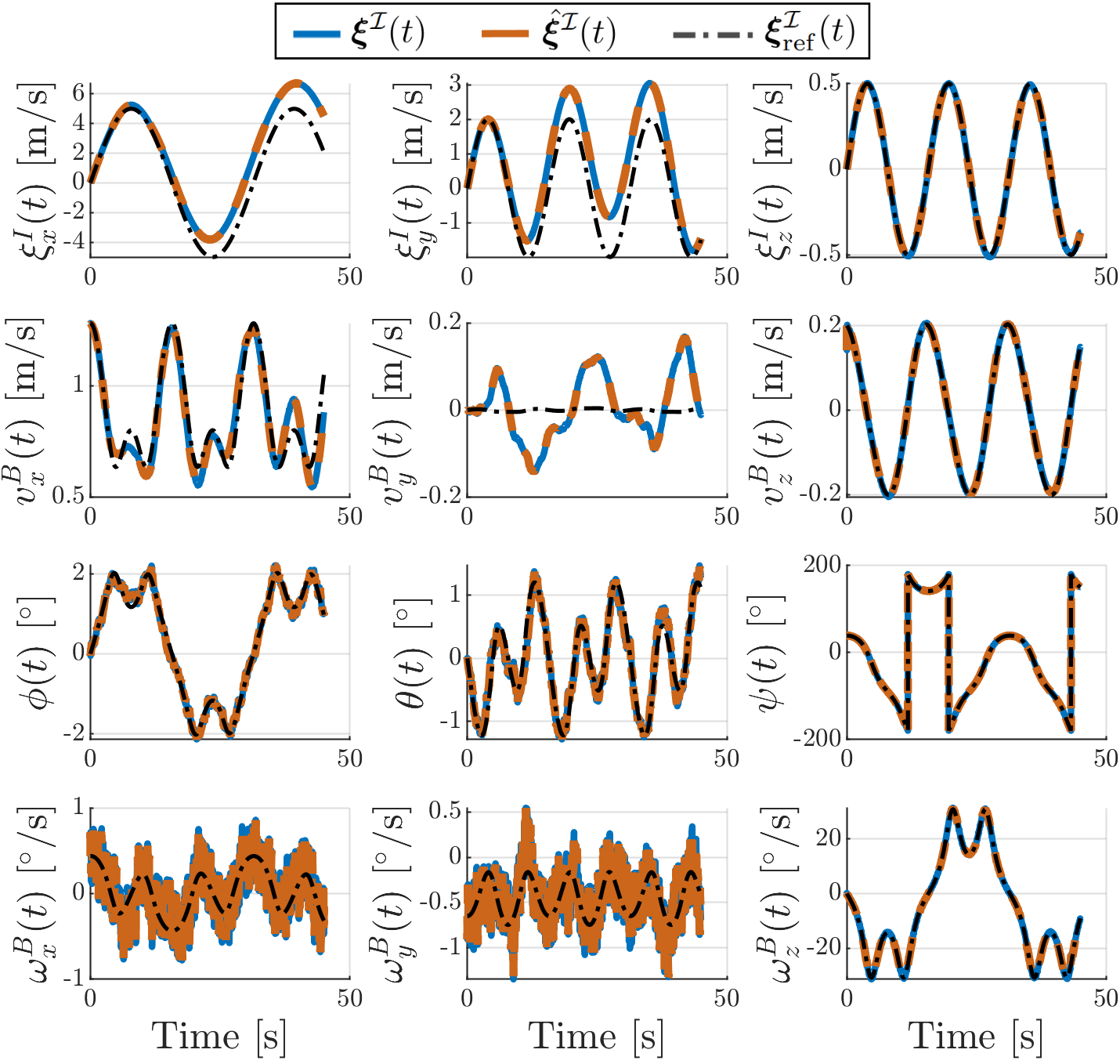}
\caption{Transient behavior of the full state-space under baseline controller $\boldsymbol{u}$ at $\mathrm{v}_{\mathrm{ref}} = 1.0\,\text{m/s}$ over a complete orbit.}
\label{fig:states_1}
\end{center}
\end{figure} 
\\
As observed, the true and estimated trajectories remain closely matched, exhibiting only a slight divergence from the reference path toward the end of the cycle. To evaluate it more closely, Fig.~\ref{fig:states_1} presents the time evolution of all 12 system states \eqref{eq:states}. This low-velocity scenario ensures high-fidelity tracking, as controller $\boldsymbol{u}$ accurately follows the trajectory kinematics with only minor tracking divergence in the lateral axes.
\\
To elucidate the correlation between the two domains, Fig.~\ref{fig:comp_est_0} evaluates the error norms for both state estimation (left, $\tilde{\boldsymbol{x}}$) and tracking (right, $\hat{\boldsymbol{\epsilon}}$) profiles. Owing to the high update rate, estimation errors remain tightly bounded throughout the cycle; however, tracking errors progressively accumulate, particularly during the sharpest turns of the trajectory.
\begin{figure}[!b]
\begin{center}
\includegraphics[width=0.485\textwidth]{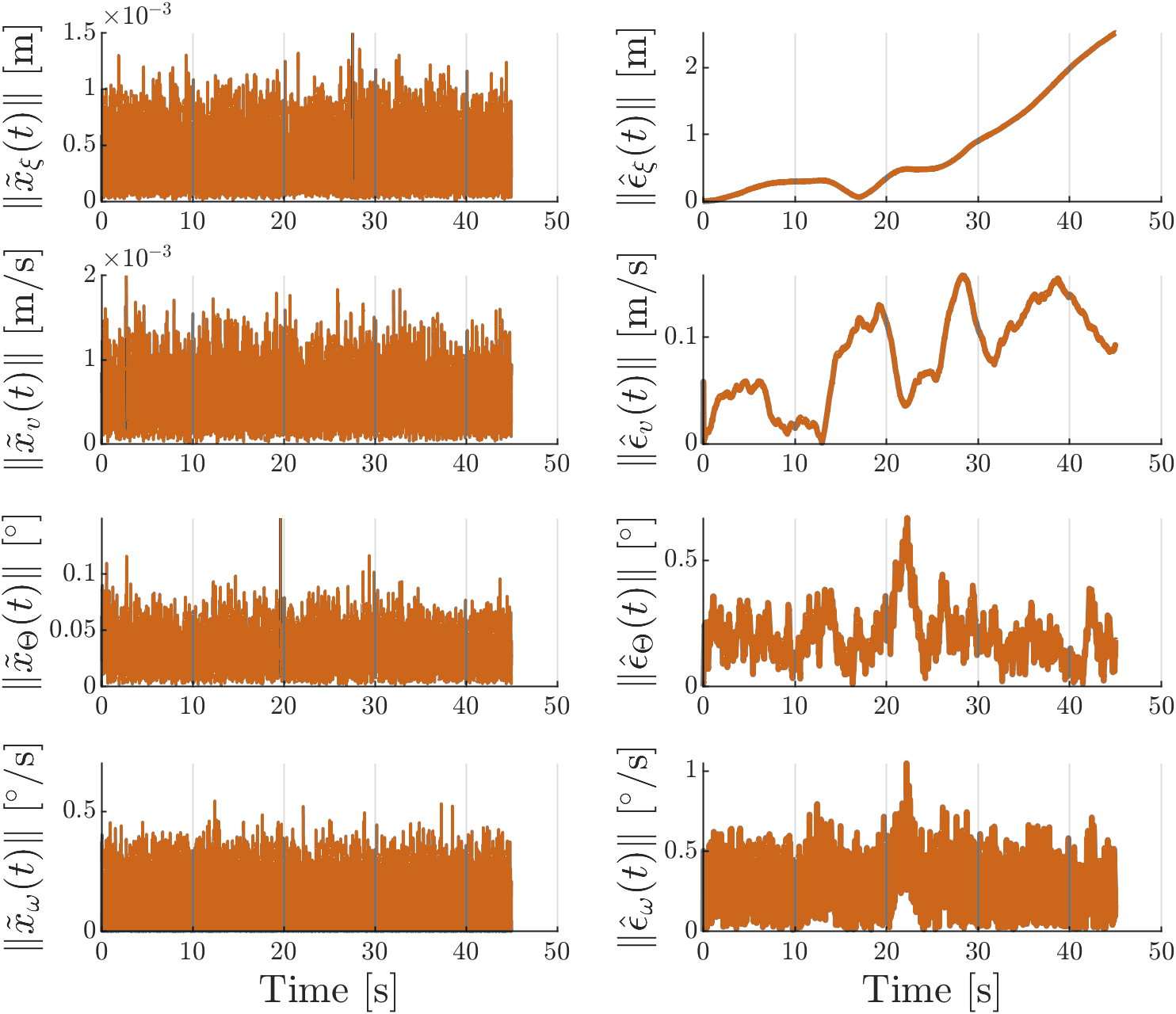}
\caption{Performance analysis ($\mathrm{v}_{\mathrm{ref}}=1.0$m/s): state estimation error (left, $\|\tilde{\boldsymbol{x}}\|$) and trajectory tracking error (right, $\|\hat{\boldsymbol{\epsilon}}\|$).}
\label{fig:comp_est_0}
\end{center}
\end{figure} 
\begin{figure}[!t]
\begin{center}
\includegraphics[width=0.49\textwidth]{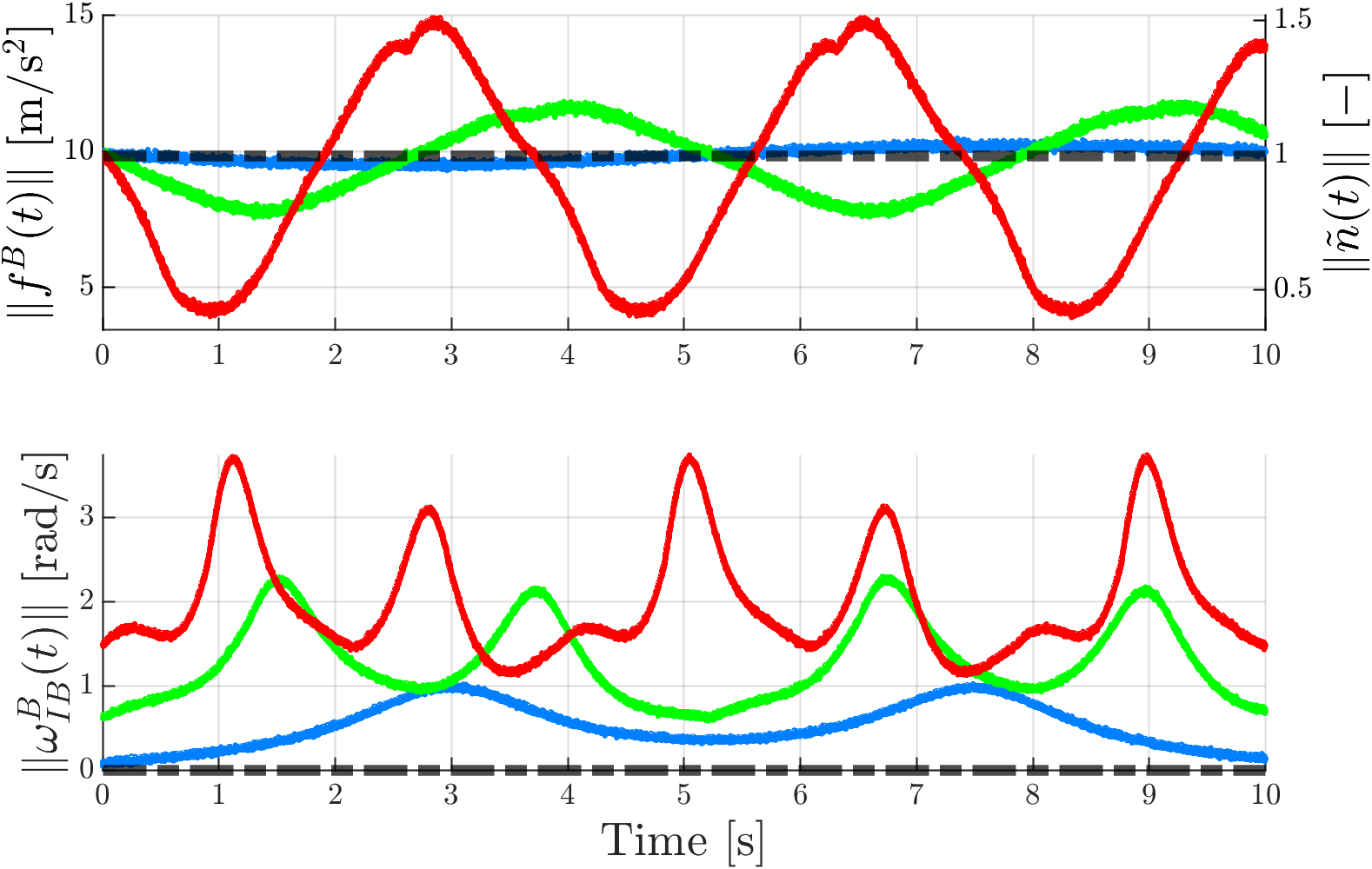}
\caption{Inertial measurements: specific forces (top) and angular rates (bottom) across velocities of $\mathrm{v}_{\mathrm{ref}} \in \{2, 5, 8 \}$ m/s.}
\label{fig:inertial}
\end{center}
\end{figure} 
\\
Fig.~\ref{fig:inertial} displays the sensed vehicle dynamics via the norm profiles of the accelerometer specific forces (top) and gyroscope angular rates (bottom). Initially, a low tracking rate ($2\text{ m/s}$, blue) yields steady profiles, with the airframe load factor remaining near nominal rest conditions ($\| \tilde{n} \| \approx 1$). 
However, as velocity increases to $5\text{ m/s}$ (green) and $8\text{ m/s}$ (red), the oscillations heighten in both amplitude and frequency, inducing severe mechanical stress that drives the peak load factor up to ($\| \tilde{n} \| \approx 1.5$). From a control perspective, these dynamics dictate trackability; while lower velocities ($\mathrm{v}_{\mathrm{ref}} \lesssim 3$~m/s) remain linearly predictable, higher speeds activate cross-couplings effects, challenging tracking via standard feedback loops.

%
\begin{figure}[!b]
\begin{center}
\includegraphics[width=0.48\textwidth]{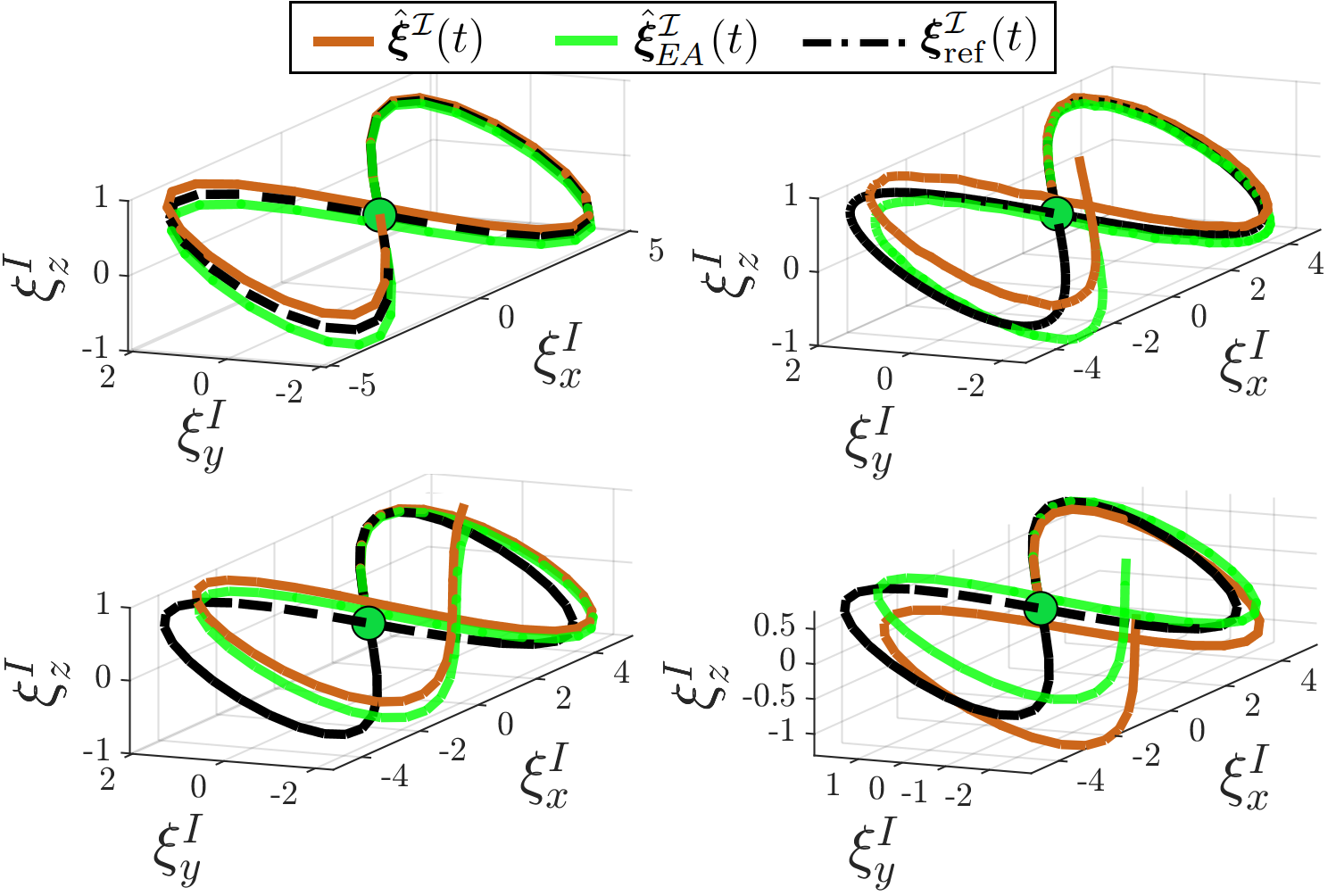}
\caption{Reference trajectory ($\boldsymbol{\xi}_{\mathrm{ref}}^{\mathcal{I}}$, black) tracking: baseline ($\boldsymbol{\xi}^{\mathcal{I}}$, brown) vs. proposed EA ($\boldsymbol{\xi}_{\mathrm{EA}}^{\mathcal{I}}$, green) controllers.}
\label{fig:lemni_sub}
\end{center}
\end{figure} 
\subsection{Comparative Analysis}
Having established the baseline performance, we now intensify the flight dynamics to evaluate how the proposed architecture \eqref{eq:u_EA} ($\boldsymbol{u}_{\textit{EA}}$, green) benchmarks against the baseline $\boldsymbol{u}$ (brown) under identical system parameters. To expose potential tracking divergence, the orbital periods are progressively increased to nominal tracking velocities of $\mathrm{v}_{\mathrm{ref}} \in \{3.27, 4.91, 7.82, 10.64\}$~m/s, corresponding to orbital periods of $T \in \{30, 20, 15, 10\}$~s. 
\\
The resulting spatial trajectories are illustrated sequentially from top-left to bottom-right in Fig.~\ref{fig:lemni_sub}.
While a baseline velocity of $\mathrm{v}_{\mathrm{ref}} = 3.27\text{ m/s}$ (top-left) is manageable for both controllers, higher rates systematically degrade tracking fidelity. This becomes severe at $\mathrm{v}_{\mathrm{ref}} = 10.64\text{ m/s}$ (bottom-right), where both architectures struggle against the escalating spatial offsets. Yet, a distinct performance gap emerges and widens monotonically with respect to $\mathrm{v}_{\mathrm{ref}}$. 
To examine it closely, we restrict our analysis to the flat outputs: translational states $\boldsymbol{\xi} = [x, y, z]^\top$ and yaw angle $\psi$; as the latter is inherently unobservable from gravity alone and drives the tracking drift.
\\
Fig.~\ref{fig:comp_est_1} evaluates these metrics at an intermediate velocity of $\mathrm{v}_{\mathrm{ref}} = 3.27\,\text{m/s}$, averaged over 10 cycles. Notably, while the estimator maintains bounded position and yaw estimation errors, the corresponding tracking errors exhibit a gradual growth over time.
\begin{figure}[!t]
\begin{center}
\includegraphics[width=0.475\textwidth]{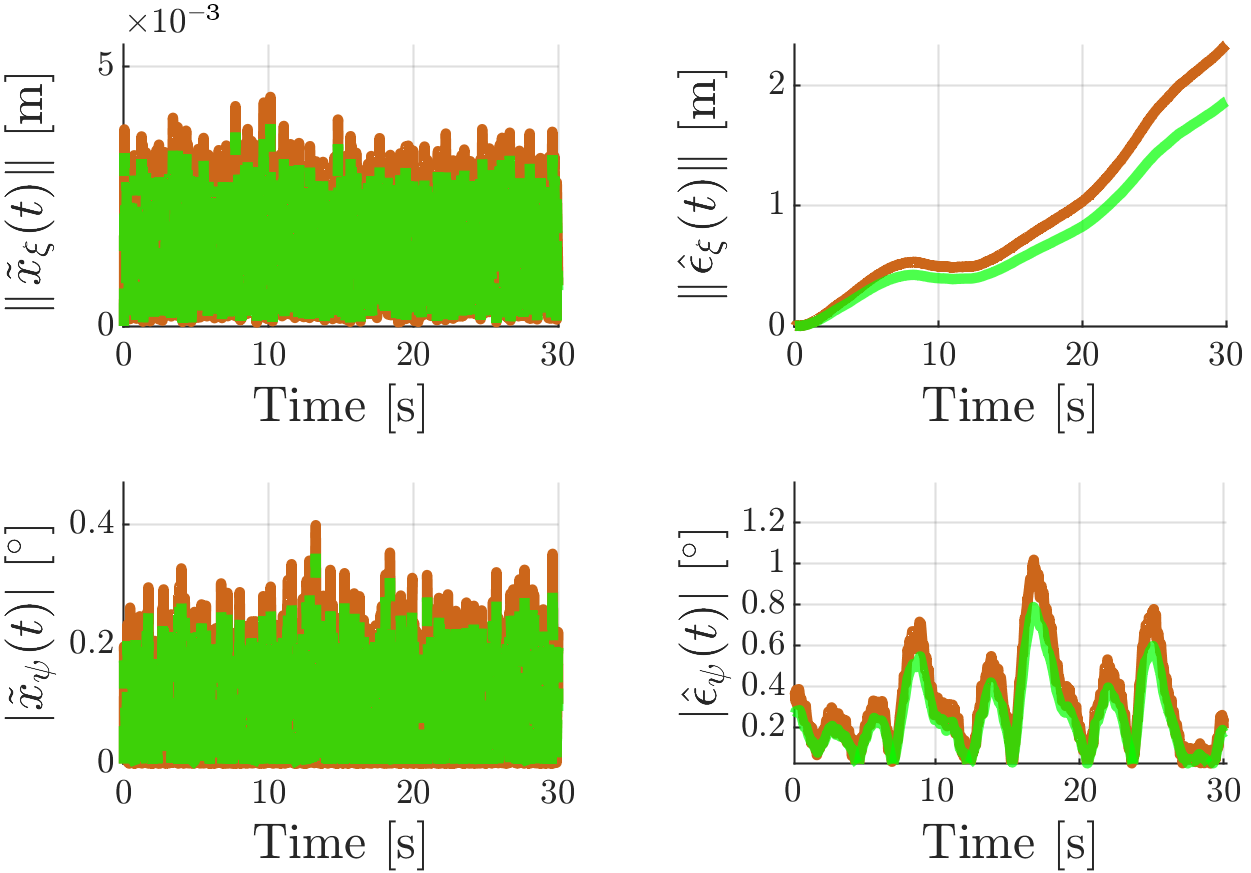}
\caption{Error comparison ($\mathrm{v}_{\mathrm{ref}}=3.27$m/s): state estimation error (left, $\|\tilde{\boldsymbol{x}}\|$) and trajectory tracking error (right, $\|\hat{\boldsymbol{\epsilon}}\|$).}
\label{fig:comp_est_1}
\end{center}
\end{figure} 
%
At this speed ($\approx 12\,\text{km/h}$), aerodynamic and cross-coupling effects are non-negligible but remain minor. This accounts for the bounded unmodeled dynamics, yielding acceptable tracking performance for both controllers.
\begin{figure}[!b]
\begin{center}
\includegraphics[width=0.495\textwidth]{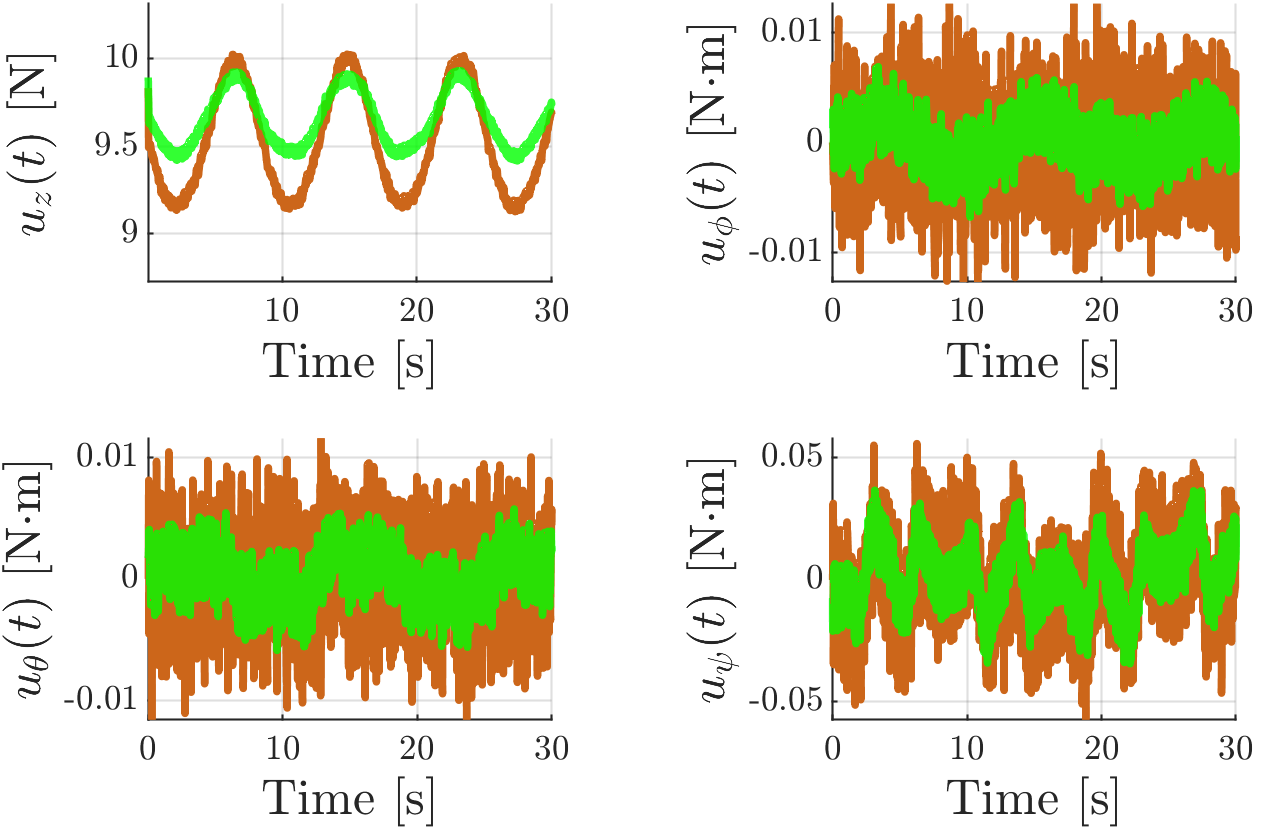}
\caption{Control comparison ($\mathrm{v}_{\mathrm{ref}}=3.27$m/s): The control inputs $\{ u_z, u_\phi, u_\theta, u_\psi \}$ mapped across the four quadrants.}
\label{fig:comp_ctrl_1}
\end{center}
\end{figure} 
%
To evaluate the underlying control allocation, Fig.~\ref{fig:comp_ctrl_1} illustrates the transient behavior of the four control inputs \eqref{eq:u_com}. As expected, the primary control authority is dedicated to the collective thrust command to counteract gravity ($\mathbb{E}[ u_z ] \approx$ mg). Conversely, the torque commands remain small, serving primarily to reorient the thrust vector via minor attitude corrections.
\\
However, a fundamentally different behavior emerges at the peak velocity of $\mathrm{v}_{\mathrm{ref}} = 10.64\,\text{m/s}$ ($\approx 38\,\text{km/h}$). The demanding requirement to complete a full 3D orbit every $10\,\text{s}$ severely strains the closed-loop system, inducing rapidly growing tracking errors $\boldsymbol{\epsilon}$ as the vehicle struggles to follow these fast kinematic variations.
\begin{figure}[!t]
\begin{center}
\includegraphics[width=0.47\textwidth]{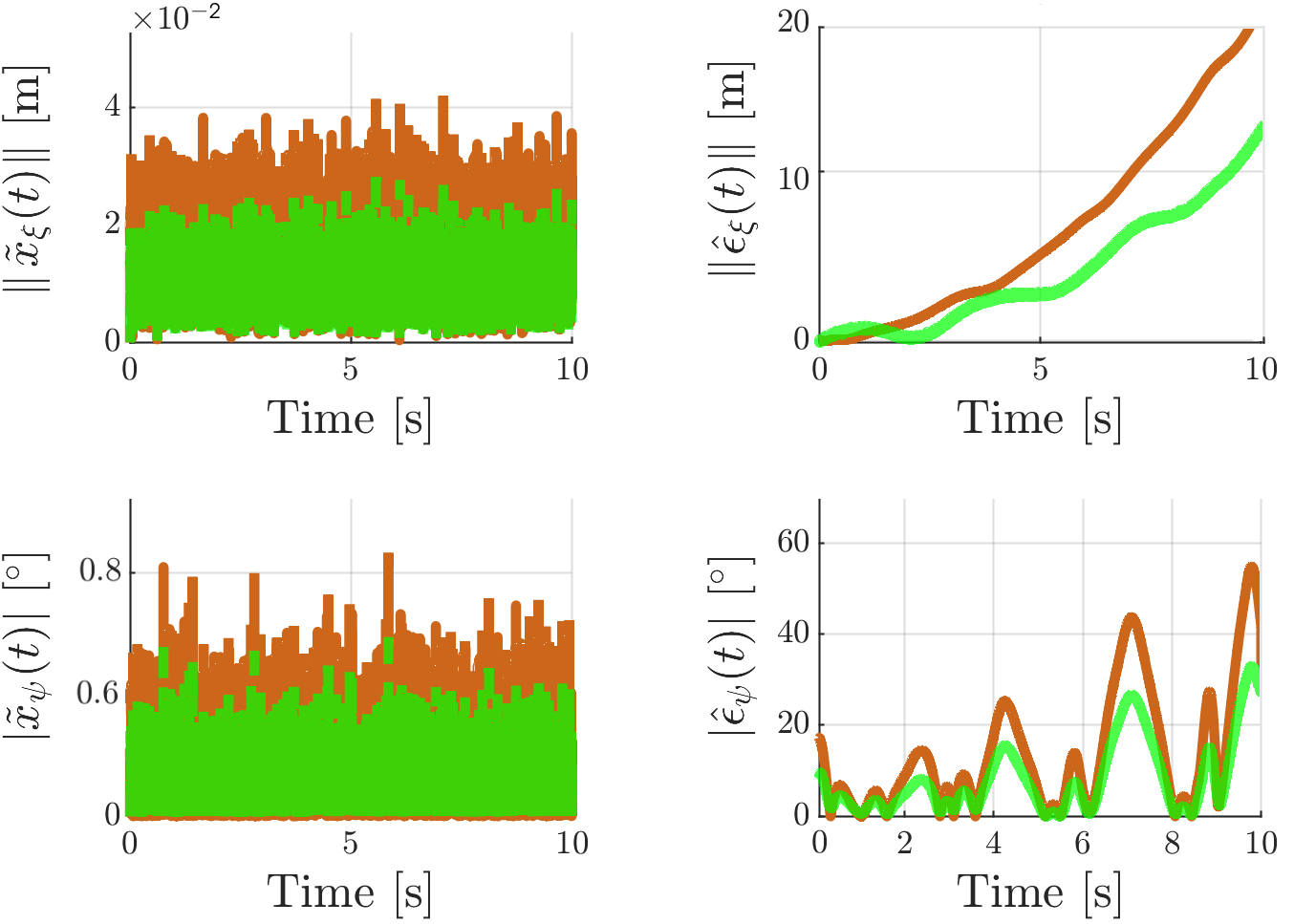}
\caption{Error comparison ($\mathrm{v}_{\mathrm{ref}}=10.64$m/s): state estimation error (left, $\|\tilde{\boldsymbol{x}}\|$) and trajectory tracking error (right, $\|\hat{\boldsymbol{\epsilon}}\|$).}
\label{fig:comp_est_2}
\end{center}
\end{figure} 
\begin{figure}[!b]
\begin{center}
\includegraphics[width=0.46\textwidth]{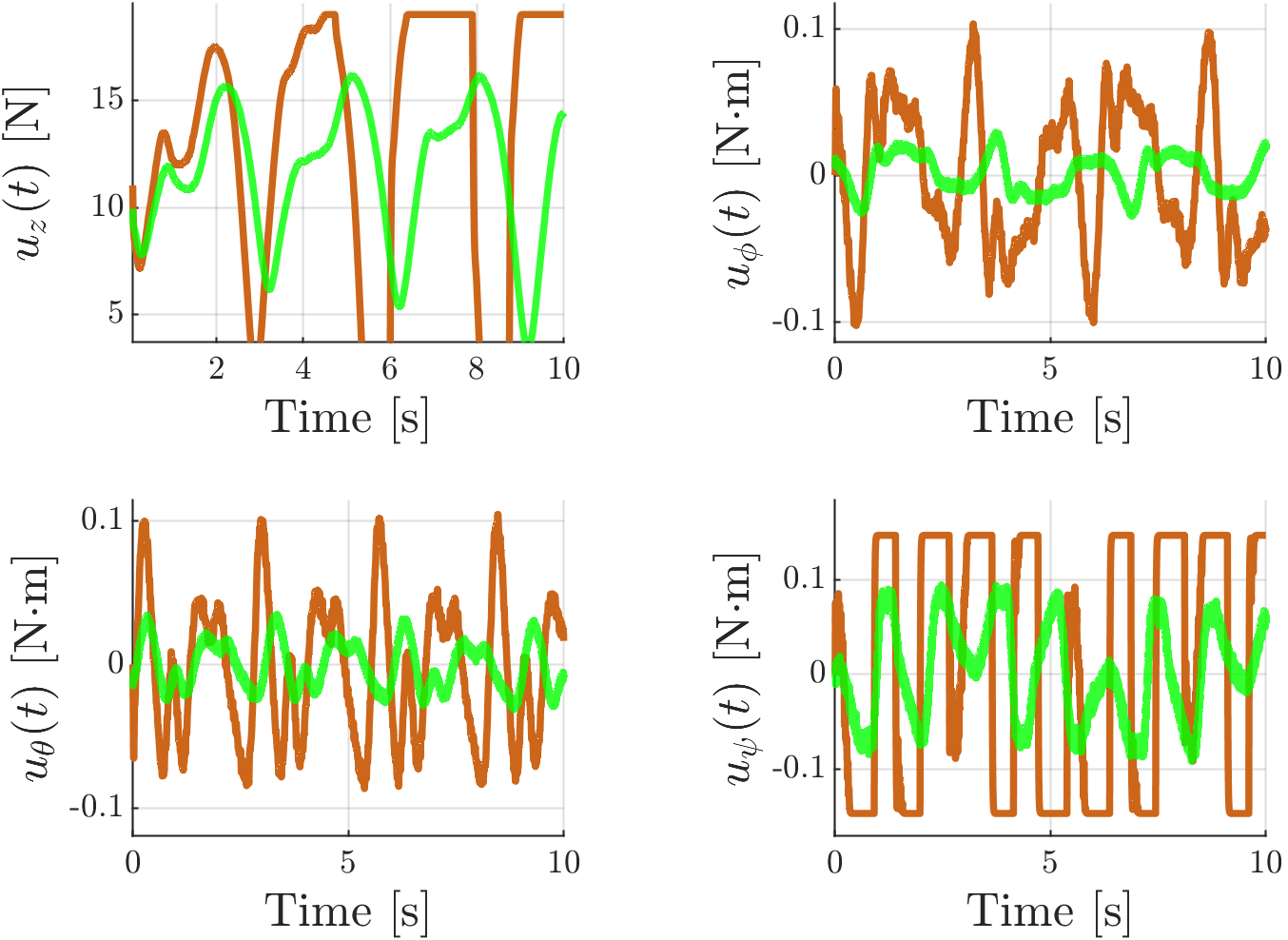}
\caption{Control comparison ($\mathrm{v}_{\mathrm{ref}}=10.64$m/s): The control inputs $\{ u_z, u_\phi, u_\theta, u_\psi \}$ mapped across the four quadrants.}
\label{fig:comp_ctrl_2}
\end{center}
\end{figure} 
\\
Fig.~\ref{fig:comp_est_2} depicts this aggressive regime, highlighting the substantial error magnitudes that accumulate over time. While both systems experience degradation, a clear performance gap widens: relative to the baseline controller ($\boldsymbol{u}$, brown), the proposed EA framework ($\boldsymbol{u}_{\textit{EA}}$, green) mitigates unobservability to maintain significantly lower tracking and estimation errors.
\\
Fig.~\ref{fig:comp_ctrl_2} illustrates how these tracking metrics manifest within the control inputs. Unlike the low-velocity scenario in Fig.~\ref{fig:comp_ctrl_1}, the collective thrust input $u_z(t)$ frequently fluctuates below $m$g, reflecting aggressive decelerations, while torque commands expand by an order of magnitude. 
\\
Crucially, while the baseline controller $\boldsymbol{u}$ frequently saturates, $\boldsymbol{u}_{\textit{EA}}$ preserves a larger control margin at a lower total expenditure, a direct consequence of decoupling the estimation-induced feedback loops. 
Table~\ref{t:err_1} quantifies these results across the tested $\mathrm{v}_{\mathrm{ref}}$ spectrum using the root mean square error (RMSE) over 10 cycles. At the lowest velocity ($1\,\text{m/s}$), performance is comparable between both methods, yielding no decisive distinction. \\
However, as flight velocities increase, a clear performance gap emerges (yellow, \textbf{\%}), demonstrating consistent improvement rates in both position and orientation of up to 40\%.
\begin{table}[!t]
\centering
\scriptsize
\caption{RMSE comparison: Estimation ($\tilde{\boldsymbol{x}}$) and tracking ($\hat{\boldsymbol{\epsilon}}$) errors are evaluated across sequential $\mathrm{v}_{\mathrm{ref}}$, contrasting the baseline (brown, $\boldsymbol{u}$) against our proposed law (green, $\boldsymbol{u}_{\textit{EA}}$). Relative improvement rates are highlighted below (yellow, \textbf{\%}).}
\renewcommand{\arraystretch}{1.5}
\begin{tabular}{|c||c|c||c|c|}
\hline 
Velocity ($\mathrm{v}_{\mathrm{ref}}$)\CC& $\| \tilde{\boldsymbol{x}}_{\boldsymbol{\xi}} \|$ [mm] \CC & $\| \tilde{\boldsymbol{x}}_{ \psi} \|$ [$^\circ$] \CC & $\| \boldsymbol{\epsilon}_{\boldsymbol{\xi}} \|$ [m] \CC & $\| \boldsymbol{\epsilon}_{\psi} \|$ [$^\circ$] \CC \\ \hline \hline  
\multirow{3}{*}{\shortstack{{1.0 m/s} \\ {(\textit{3.6 km/h})} }} &  1.92 \CG & 0.12 \CG & 0.12 \CG & 0.24 \CG \\ \cline{2-5}
 & 2.10 \CB & 0.11 \CB & 0.13 \CB & 0.23 \CB  \\ \cline{2-5}
 & \textbf{-9.37 \%} \CY & \textbf{4.17 \%} \CY & \textbf{-8.33 \%} \CY & \textbf{2.92 \%} \CY  \\ \hline \hline
\multirow{3}{*}{\shortstack{{3.27 m/s} \\ {(\textit{11.8 km/h})} }} & 3.86 \CG & 0.22 \CG& 1.12 \CG& 0.44 \CG \\ \cline{2-5}
 & 3.42 \CB & 0.19 \CB & 0.98 \CB & 0.40 \CB  \\ \cline{2-5}
 & \textbf{11.4 \%} \CY & \textbf{13.6 \%} \CY &  \textbf{12.5 \%} \CY & \textbf{9.09 \%} \CY  \\ \hline \hline 
\multirow{3}{*}{\shortstack{{4.91 m/s} \\ {(\textit{17.7 km/h})} }} & 6.45 \CG & 0.31 \CG & 2.54 \CG & 1.95 \CG \\ \cline{2-5}
 & 5.24 \CB & 0.26 \CB & 2.06 \CB & 1.62 \CB  \\ \cline{2-5}
 & \textbf{18.8 \%} \CY & \textbf{16.1 \%} \CY & \textbf{18.9 \%} \CY & \textbf{16.9 \%} \CY  \\ \hline \hline
\multirow{3}{*}{\shortstack{{7.82 m/s} \\ {(\textit{28.1 km/h})} }} & 16.2 \CG & 0.47 \CG & 6.88 \CG & 7.14 \CG \\ \cline{2-5}
 & 11.5 \CB & 0.37 \CB & 4.71 \CB & 4.82 \CB  \\ \cline{2-5}
 & \textbf{29.0 \%} \CY & \textbf{21.3 \%} \CY & \textbf{31.5 \%} \CY & \textbf{32.5 \%} \CY  \\ \hline \hline
\multirow{3}{*}{\shortstack{{10.64 m/s} \\ {(\textit{38.3 km/h})} }} & 31.1 \CG & 0.63 \CG& 12.3 \CG& 18.7 \CG \\ \cline{2-5}
 & 18.8 \CB & 0.48 \CB & 7.81 \CB & 11.2 \CB  \\ \cline{2-5}
  & \textbf{39.5 \%} \CY & \textbf{23.8 \%} \CY &  \textbf{36.5 \%} \CY & \textbf{40.1 \%} \CY  \\ \hline
\end{tabular} \label{t:err_1}
\end{table} 
\begin{figure}[!b]
\begin{center}
\includegraphics[width=0.48\textwidth]{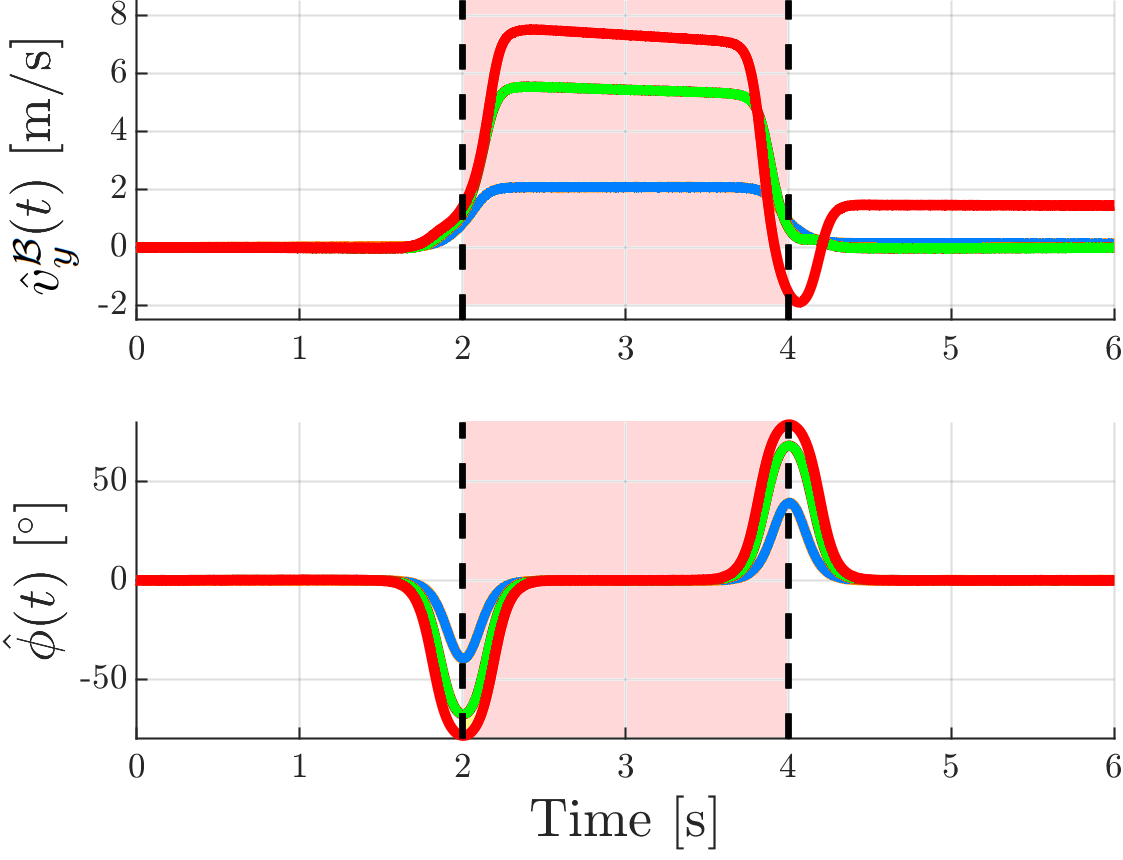}
\caption{Time evolution of lateral velocity (top) and roll angle (bottom) under pulsed 2~s step ($v_{\mathrm{ref},y}^{\mathcal{I}} \in \{2, 5, 8\}$ m/s, red).}
\label{fig:EA_aware_1}
\end{center}
\end{figure} 
\subsection{Awareness Analysis} 
To assess how estimation confidence maps to system awareness, we evaluate a demanding flight maneuver designed to challenge the controller-estimator coupling. 
\\
As illustrated in Fig.~\ref{fig:EA_aware_1}, the vehicle departs from a steady-state hover at $t=2$~s via an abrupt lateral step velocity command, $v_{\mathrm{ref},y}^{\mathcal{I}}$, lasting for $2$~s ($2 \le t \le 4$~s, pale red), before returning to a stable, stationary hover.
Iterating this scenario across $v_{\mathrm{ref},y}^{\mathcal{I}} \in \{2, 5, 8\}$~m/s, in blue, green, and red, respectively, shows that the lateral velocity (top) and roll angle (bottom) amplitudes scale proportionally with the command magnitude. 
Crucially, this translation-attitude coupling induces a non-minimum phase undershoot inherent to underactuated aircraft. Next, to observe how the EA mechanism \eqref{eq:EA_act} responds to these dynamics, the error covariance determinant, $\det(\overline{\mathbf{P}})$, is plotted in Fig.~\ref{fig:EA_2} to map its structural expansions over time.
\begin{figure}[!t]
\begin{center}
\includegraphics[width=0.48\textwidth]{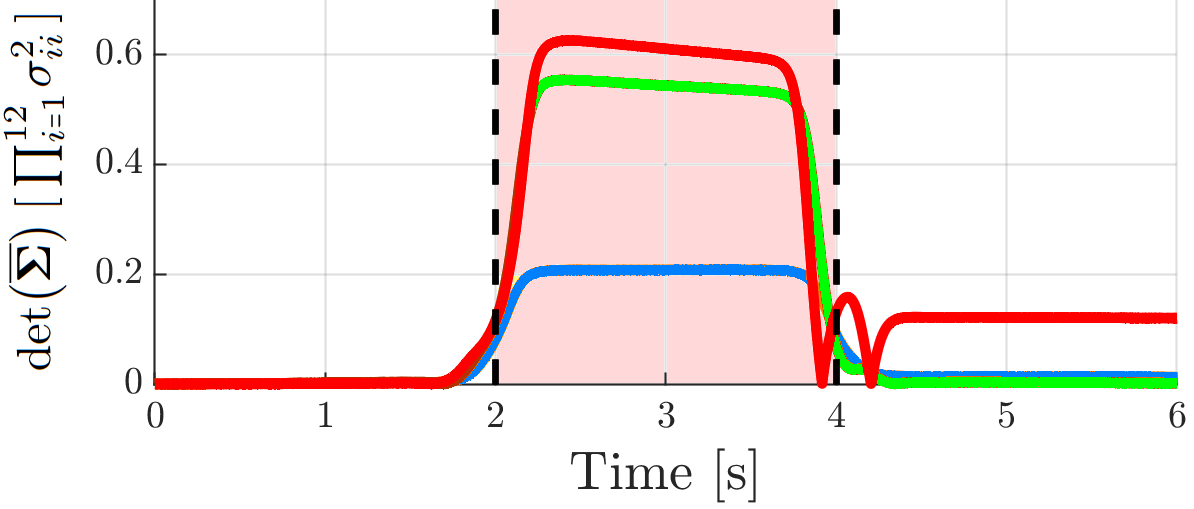}
\caption{Time evolution of $\det(\overline{\mathbf{P}})$ showcasing uncertainty expansions under pulsed 2~s step ($v_{\mathrm{ref},y}^{\mathcal{I}} \in \{2, 5, 8\}$ m/s, red).}
\label{fig:EA_2}
\end{center}
\end{figure} 
\\
As expected, uncertainty remains near zero during initial hover. Upon applying the step command, $\det(\overline{\mathbf{P}})$ spikes sharply in proportion to $v_y$ before returning to baseline once the transient maneuver subsides ($t > 4$~s). These results align with the EA triggering mechanism: during mild dynamics, the awareness gate matrix $\boldsymbol{\Delta}(\boldsymbol{\eta})$ vanishes as $\overline{\mathbf{P}} \to \mathbf{I}$, recovering nominal INDI control. Conversely, during intense maneuvers, the off-diagonal elements of $(\mathbf{I} - \overline{\mathbf{P}})$ flare up, dynamically dampening estimation-induced closed-loop distortions.

\subsection{Frequency-Domain Analysis}
To complement the previous evaluation, this section extends the comparison between $\boldsymbol{u}$ and $\boldsymbol{u}_{\textit{EA}}$ into the frequency domain across the entire forward velocity spectrum $v_x$. To apply the Laplace transform ($s=\sigma + j \omega$), a coordinated turn is selected as the trim condition; this operating point actively subjects the vehicle to severe rigid-body cross-couplings and aerodynamic effects while remaining linearizable within the flight envelope. To that end, the closed-loop (CL) system state-space matrices were swept, and their respective poles, singular values, and principal phase angles were extracted across a forward velocity range of $v_x \in [2, 16]$ m/s.
%
%
Fig.~\ref{fig:root_locus} illustrates the resulting pole migration in response to a $v_x$ sweep (colorbar), highlighting how the loci for the baseline ($\boldsymbol{u}$, $\square$) and the EA law ($\boldsymbol{u}_{\textit{EA}}$, $\boldsymbol{\times}$) traverse the complex $s$-plane.
\begin{figure}[!t]
\begin{center}
\includegraphics[width=0.495\textwidth]{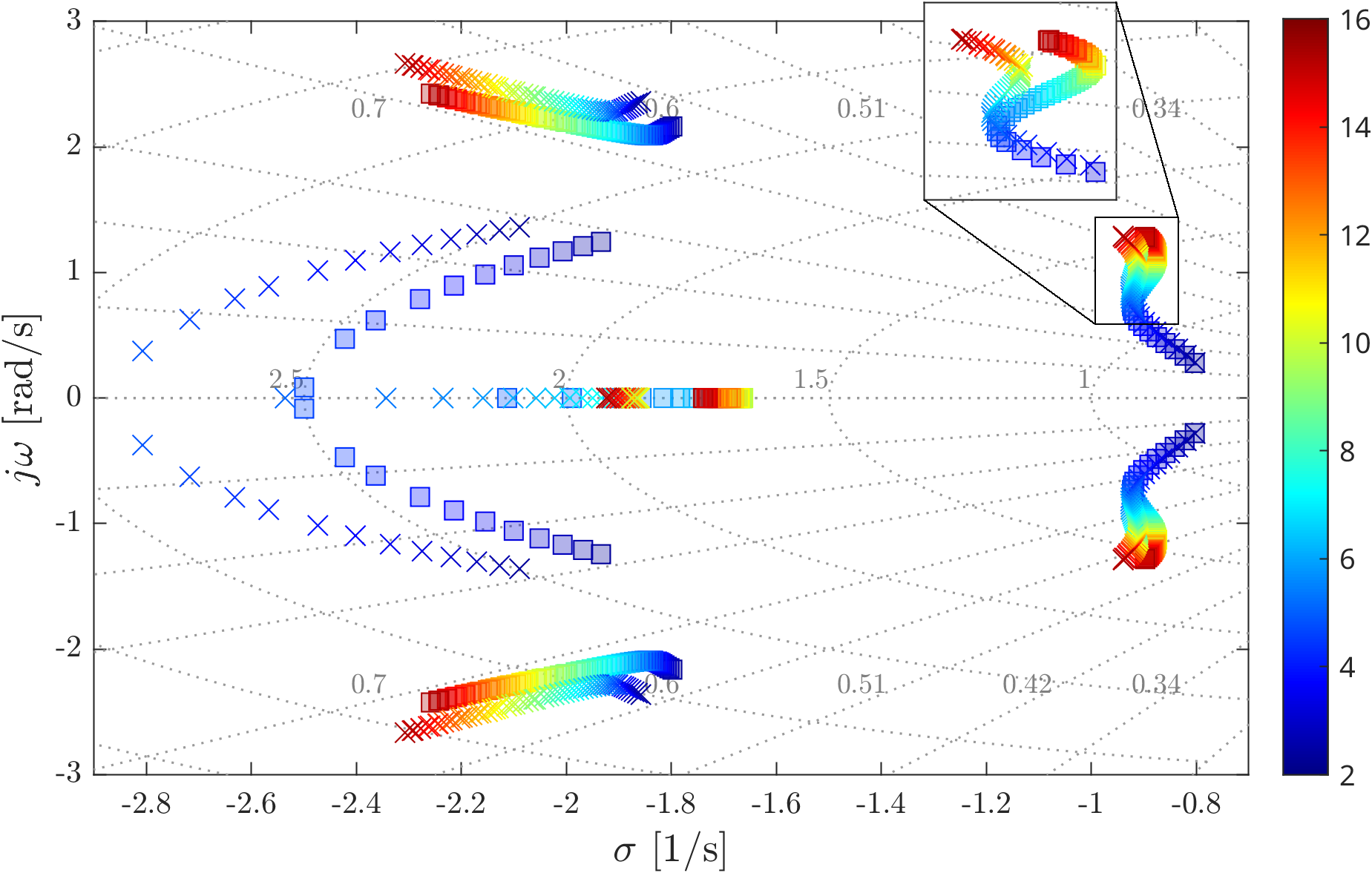}
\caption{Parametric root locus: CL pole migration of the baseline controller ($\boldsymbol{u}$, $\square$) with the proposed EA law ($\boldsymbol{u}_{\textit{EA}}$, $\boldsymbol{\times}$) across a forward velocity range of $v_x \in [2, 16]$ m/s.}
\label{fig:root_locus}
\end{center}
\end{figure} 
\\
To reduce clutter and streamline interpretation, the complex plane is bounded to low-to-medium frequencies, where the primary aircraft dynamic modes reside and dictate short-term stability. On the right side, the dominant low-frequency poles ($\sigma \approx 0.9$ 1/s) capture a heavily coupled lateral-directional mode, characteristic of the Dutch roll dynamics modified by the tight coordinated turn tracking requirements. Conversely, moving leftward, the higher-damping region ($\sigma \lesssim -1.8$ 1/s) accommodates two distinct dynamic structures: a fast, oscillatory short-period attitude mode ($\omega \neq 0$) driven by the high-gain feedback control law, and highly damped, non-oscillatory modes ($\omega = 0$) tracking along the real axis, which represent the heavily suppressed spiral and roll subsidence behaviors.
\begin{table*}[!b]
\centering
\scriptsize
\caption{Frequency-domain stability and tracking metrics of the closed-loop system: comparison between the baseline (brown, $\boldsymbol{u}$) and the proposed EA law (green, $\boldsymbol{u}_{\textit{EA}}$) controllers, with corresponding relative improvement rates highlighted below (yellow, \textbf{\%}).}
\renewcommand{\arraystretch}{1.8}
\begin{tabular}{|c||c|c|c||c|c||c|c|}
\hline 
\makecell{Forward\\velocity ($v_x$)} \CC & \makecell{DC Gain\\ $[$dB$]$} \CC & \makecell{Cutoff Freq. \\$\omega_c$ [rad/s]} \CC & \makecell{Time Constant \\ $\tau$ [s]}\CC & \makecell{Gain Crossover\\ ($\omega_{\text{gc}}$) [rad/s]}\CC & \makecell{Phase\\Margin [$^\circ$]} \CC & \makecell{Phase Crossover\\ ($\omega_{\text{pc}}$) [rad/s]} \CC & \makecell{Gain\\Margin [dB]} \CC \\ \hline \hline  
\multirow{3}{*}{\shortstack{{2 m/s} \\ {(\textit{7.2 km/h})} }} & 28.12 \CG & 1.62 \CG & 0.617 \CG & 46.20 \CG & 54.2 \CG & 82.5 \CG & 7.8 \CG \\ \cline{2-8}
 & 28.15 \CB & 1.65 \CB & 0.606 \CB & 46.55 \CB & 54.8 \CB & 83.1 \CB & 8.0 \CB  \\ \cline{2-8}
 & \textbf{0.11 \%} \CY & \textbf{1.85 \%} \CY &  \textbf{1.78 \%} \CY & \textbf{0.76 \%} \CY  &  \textbf{1.11 \%} \CY & \textbf{0.73 \%} \CY & \textbf{2.56 \%} \CY  \\ \hline \hline 
\multirow{3}{*}{\shortstack{{6.67 m/s} \\ {(\textit{24.0 km/h})} }} & 28.54 \CG & 1.84 \CG & 0.543 \CG & 47.15 \CG & 51.5 \CG & 80.4 \CG & 7.1 \CG \\ \cline{2-8}
 & 28.90 \CB & 2.12 \CB & 0.472 \CB & 48.30 \CB & 53.9 \CB & 82.8 \CB & 7.9 \CB  \\ \cline{2-8}
 & \textbf{1.26 \%} \CY & \textbf{15.22 \%} \CY &  \textbf{13.08 \%} \CY & \textbf{2.44 \%} \CY  &  \textbf{4.66 \%} \CY & \textbf{2.99 \%} \CY & \textbf{11.27 \%} \CY  \\ \hline \hline 
\multirow{3}{*}{\shortstack{{11.3 m/s} \\ {(\textit{40.8 km/h})} }} & 29.21 \CG & 2.11 \CG & 0.474 \CG & 48.92 \CG & 46.3 \CG & 76.2 \CG & 6.2 \CG \\ \cline{2-8}
 & 29.85 \CB & 2.75 \CB & 0.364 \CB & 50.45 \CB & 52.8 \CB & 81.9 \CB & 7.8 \CB  \\ \cline{2-8}
  & \textbf{2.19 \%} \CY & \textbf{30.33 \%} \CY &  \textbf{23.21 \%} \CY & \textbf{3.13 \%} \CY  &  \textbf{14.04 \%} \CY & \textbf{7.48 \%} \CY & \textbf{25.81 \%} \CY  \\ \hline \hline 
\multirow{3}{*}{\shortstack{{16 m/s} \\ {(\textit{57.6 km/h})} }} & 29.45 \CG & 2.32 \CG & 0.431 \CG & 50.12 \CG & 38.5 \CG & 68.4 \CG & 4.9 \CG \\ \cline{2-8}
 & 30.12 \CB & 3.24 \CB & 0.309 \CB & 52.35 \CB & 51.2 \CB & 80.5 \CB & 7.6 \CB  \\ \cline{2-8}
  & \textbf{2.28 \%} \CY & \textbf{39.66 \%} \CY &  \textbf{28.31 \%} \CY & \textbf{4.45 \%} \CY  &  \textbf{32.99 \%} \CY & \textbf{17.69 \%} \CY & \textbf{55.10 \%} \CY  \\ \hline
\end{tabular} \label{t:err_2}
\end{table*} 
%
Despite the superficially similar trajectories across the velocity sweep, the EA control law offers two critical advancements. First, $\boldsymbol{u}_{\textit{EA}}$ provides superior transient damping; as $v_x$ increases, its medium-frequency branches push significantly deeper into the left-half plane, achieving a peak attenuation rate of $\sigma_{\textit{EA}} \approx -2.8$ 1/s compared to the baseline's limit of $\sigma \approx -2.4$ 1/s. This guarantees faster suppression of transient attitude oscillations. Second, $\boldsymbol{u}_{\textit{EA}}$ enhances high-velocity modal robustness, maintaining a more stable CL damping profile ($\zeta_{\textit{EA}} > \zeta$). 
\begin{figure}[!t]
\begin{center}
\includegraphics[width=0.455\textwidth]{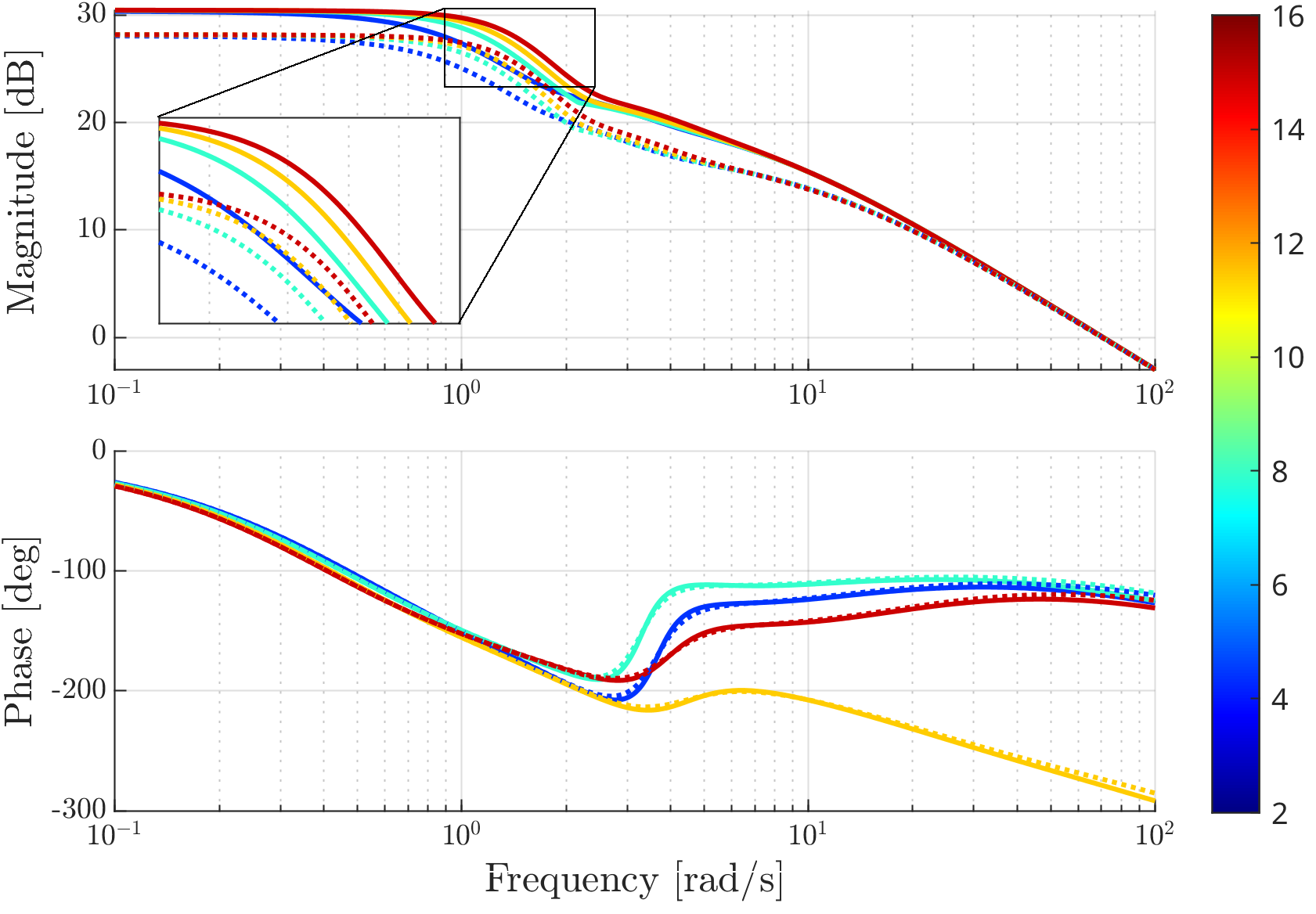}
\caption{Parametric Bode plot: CL frequency response comparison of the baseline ($\boldsymbol{u}$, dashed) and the proposed EA law ($\boldsymbol{u}_{\textit{EA}}$, solid) evaluated at four distinct forward velocities.}
\label{fig:bode_plot}
\end{center}
\end{figure} 
\\
While both controllers track closely at low speeds (blue markers), at the upper envelope bound of $16$ m/s (red markers), the EA poles remain bounded lower on the imaginary axis, effectively preventing the aircraft's coupled dynamics from becoming excessively oscillatory under high dynamic pressures.
Lastly, the corresponding CL frequency response is analyzed to evaluate tracking performance and robust stability. 
\\
Fig.~\ref{fig:bode_plot} illustrates the multivariable Bode diagram, plotting the maximum singular value magnitude (top) and principal phase angles (bottom) across the frequency spectrum. 
While the differences between the control laws appear minor from a macroscopic view, they prove highly consequential for the tracking mission. By mitigating severe cross-couplings and enforcing higher low-frequency transmission gains, the EA law ($\boldsymbol{u}_{\textit{EA}}$, solid lines) maintains higher amplification ($\approx 30\text{ dB}$) and pushes the CL tracking bandwidth upward. 
This higher bandwidth directly translates to sharper reference tracking and accelerated disturbance rejection across the velocity spectrum. 
Concurrently, the favorable phase profile maintained by the EA law near the crossover region increases the CL modal damping ratio. This directly mirrors the root locus findings, ensuring a marked reduction in transient oscillations and rendering the vehicle significantly less sensitive to sensor and control latencies—phenomena that typically exacerbate during high-velocity maneuvers, making this robust behavior crucial for the operating conditions at hand.
Table~\ref{t:err_2} summarizes the key frequency-domain metrics extracted from this chapter's analysis. As observed, the proposed EA law consistently improves upon the baseline performance across the entire flight envelope; while performance is nearly identical at low velocities, the EA framework demonstrates a progressive resilience to high dynamic pressures, culminating in substantial enhancements of up to $55.10\%$ in stability margins and nearly $40\%$ in tracking bandwidth at the $16\text{ m/s}$ envelope bound.

\section{Conclusion} \label{sec:conc}
In this paper, we formalized the structural estimation-control coupling intrinsic to underactuated nonlinear systems. While high-velocity tracking pushes aerial platforms to their operational limits, aggressive maneuvering exacerbates underactuated dynamics, heightening sensitivity to state drift and sensor degradation. To address this, we leveraged Lyapunov analysis to characterize the closed-loop stability bounds under these coupled dynamics. We then introduced an EA control framework designed to systematically isolate estimation-induced feedback loops. 
\\
Extensive quadrotor flight experiments along complex 3D trajectories at speeds up to $57.6\,\text{km/h}$ validate that the proposed architecture preserves closed-loop stability where conventional, separation-principle-based designs diverge. Quantitative frequency- and time-domain evaluations confirm that the EA control law expands tracking bandwidth by 39\% and improves stability margins by up to 55\%, ensuring robust performance across the entire velocity spectrum. As autonomous systems transition toward full operational independence under severe sensor constraints, the proposed EA framework offers a mathematically rigorous paradigm for agile, safe flight control in unpredictable, high-rate environments.


\bibliographystyle{unsrt}
\bibliography{Ref}

\end{document}